\documentclass[letterpaper,twocolumn,10pt]{article}
\usepackage{amsfonts,amsmath,amsthm,amssymb}
\usepackage{usenix,epsfig}
\usepackage{url}
\usepackage{hyperref}
\usepackage{tikz,pgfplots}
\usepackage{breqn}

\usepackage[T1]{fontenc}

\newtheorem{theorem}{Theorem}

\newcommand{\squishlist}{
  \begin{list}{$\bullet$}
   {
     \setlength{\itemsep}{0pt}
     \setlength{\parsep}{0pt}
     \setlength{\topsep}{0pt}
     \setlength{\partopsep}{0pt}
     \setlength{\leftmargin}{1.5em}
     \setlength{\labelwidth}{1em}
     \setlength{\labelsep}{0.5em} } 
}

\newcommand{\squishend}{
   \end{list} 
}

\begin{document}

\date{}

\title{Tempo: Robust and Self-Tuning Resource Management in
  Multi-tenant Parallel Databases}

\author{
{\rm Zilong Tan}\\
{\rm Duke University}\\
{\rm ztan@cs.duke.edu}
\and
{\rm Shivnath Babu}\\
{\rm Duke University}\\
{\rm shivnath@cs.duke.edu}
} 

\maketitle



\subsection*{Abstract}
Multi-tenant database systems have a component called the {\em
  Resource Manager, or RM} that is responsible for allocating
resources to tenants. RMs today do not provide direct support for
performance objectives such as: ``Average job response time of tenant
A must be less than two minutes'', or ``No more than 5\% of tenant B's
jobs can miss the deadline of 1 hour.''  Thus, DBAs have to tinker
with the RM's low-level configuration settings to meet such
objectives.  We propose a framework called {\em Tempo} that brings
{\em simplicity}, {\em self-tuning}, and {\em robustness} to existing
RMs.  Tempo provides a simple interface for DBAs to specify
performance objectives declaratively, and optimizes the RM
configuration settings to meet these objectives.  Tempo has a solid
theoretical foundation which gives key robustness guarantees. We
report experiments done on Tempo using production traces of
data-processing workloads from companies such as Facebook and
Cloudera.  These experiments demonstrate significant improvements in
meeting desired performance objectives over RM configuration settings
specified by human experts.

\section{Introduction}
\label{sec:intro}

Many enterprises today run multi-tenant database systems on large
shared-nothing clusters. Examples of such systems include parallel SQL
database systems like RedShift \cite{redshift}, Teradata
\cite{teradata}, and Vertica \cite{vertica}, Hadoop/YARN running SQL
and MapReduce workloads, Spark running on Mesos \cite{Hindman11} or
YARN \cite{Vavil13}, and many others.  Meeting the performance goals
of business-critical workloads (popularly called {\em service-level
  objectives}, or {\em SLOs}) while achieving high resource
utilization in multi-tenant database systems has become more important
and challenging than ever.

The problem of handling many (often in 1000s) small and independent
databases on a multi-tenant database Platform-as-a-Service (usually
called {\em PaaS} or {\em DBaaS}) has received considerable attention
in recent years \cite{Xiong11,Nara13,slo11,viv13,Sud13}.  That is not
the problem we focus on in this paper. Our focus is on handling fewer,
but much ``bigger'', tenants who process very large amounts of data on
a shared-nothing cluster that is usually run within an enterprise.
Hadoop, Spark, Teradata, Vertica, etc., are typically run in such
settings.

These multi-tenant database systems each have a component---commonly
referred to as the {\em Resource Manager (RM)} (also sometimes called
{\em Workload Manager})---that is responsible for allocating resources
to tenants. Most widely deployed RMs like YARN and Mesos do not
support SLOs. Instead, they rely on the Database Administrator (DBA)
to ``guesstimate'' answers to questions such as: ``How much resources
are needed to complete this job before its deadline?'' Then, DBAs have
to translate their answers into low-level configuration settings in
the RM.  This process is brittle and increasingly hard as workloads
evolve, data and cluster sizes change, and new workloads are added.
Thus, techniques have been proposed in the literature to support
specific SLOs such as deadlines \cite{Curino14,Li14,Ferg12,ARIA}, fast
job response times \cite{Boutin14,Curino14,Grandl14,Ous13}, high
resource utilization \cite{Corona,Boutin14,Curino14}, scalability
\cite{Corona,Schw13,Zhang14}, and transparent failure recovery
\cite{Zhang14}.

In this paper, we present a framework called {\em Tempo} that brings
three properties to existing RMs: {\em simplicity}, {\em self-tuning},
and {\em robustness}.  First, Tempo provides a simple interface for
DBAs to specify SLOs declaratively. Thus, Tempo enables the RM to be
made aware of SLOs such as: ``Average job response time of tenant A
must be less than two minutes'', and ``No more than 5\% of tenant B's
jobs can miss the deadline of 1 hour.''  Second, Tempo constantly
monitors the SLO compliance in the database, and adaptively optimizes
the RM configuration settings to maximize SLO compliance.  Third,
Tempo has a solid theoretical foundation which gives five critical
robustness guarantees:

\squishlist
\item[1)] Tempo's optimization and modeling algorithms account for the
  noisy nature of production database systems.
\item[2)] Tempo's optimization algorithm converges provably to a
  Pareto-optimal RM configuration given that satisfying multiple
  tenant SLOs is a multi-objective optimization problem.
\item[3)] When all SLOs cannot be satisfied---which is common in busy
  database systems---Tempo guarantees max-min fairness over SLO
  satisfactions \cite{mace15}.
\item[4)] Tempo adapts to workload patterns and variations.
\item[5)] Tempo reduces the risk of major performance regression while
  being applied to production database systems.
\squishend

\noindent We have implemented Tempo as a drop-in component in the RMs
used by multi-tenant databases running on Hadoop and Spark.  We report
experiments done using production traces of data-processing workloads
from companies such as Facebook and Cloudera.  These experiments
demonstrate significant improvements in meeting the SLOs over the
original RMs used in five real-life scenarios.  For example, Tempo can
reduce the average job response time by 50\% for best-effort workloads
and increase resource utilization by 15\%, without hurting the
deadline-driven workloads.

\section{Production System Experiences}
\label{sec:experiences}

Tempo's design was motivated by our observations from several large
production database systems.  While designing Tempo, we analyzed
workload traces from three companies each of which runs multi-tenant
database systems on large clusters. Two of these systems run on 600+
nodes while the other runs on about 150 nodes. (While all three are
well-known companies, we cannot share their names due to legal
restrictions.) We talked to business analysts, application developers,
team managers, and DBAs in these teams to understand the SLOs that
they need to meet and the challenges they face in resource
management. From all our interviews, the following emerged as the top
concerns:

\squishlist
\item Concern A: Deadline-based workloads and best-effort workloads
  have to be supported on the same database system.
\item Concern B: Repeatedly-run jobs often have unpredictable
  completion times.
\item Concern C: Resource utilization was lower than expected.
\item Concern D: Resource allocation does not adapt automatically to
  the patterns and variations in the workloads.
\squishend

\noindent To elaborate on these four concerns, we will use one of the
three companies---henceforth, referred to as Company ABC---which is a
real-life company that runs a multi-tenant database system on a
700-node Hadoop cluster with over 30 Petabytes of data.

\subsection{Concern A}
\label{sec:concern-a}

Company ABC has three types of users who generate database workloads.
Business Intelligence (BI) analysts and Data Scientists predominantly
do exploratory analysis on the data. Engineers develop and maintain
recurring jobs that run on the database.  One such category of jobs is
Extract-Transform-Load (ETL) which brings new data into the system.
Each job goes through many runs in a development phase on the cluster
before being certified to run as a production job. Thus, the system
supports both development and production runs of jobs.

Distinct workloads from these users form the {\em tenants} in the
multi-tenant system.  Table~\ref{tbl:rf_users} shows the six tenants
at Company ABC and their distinct workload characteristics. (The
experimental evaluation section gives more fine-grained details of
these workloads.)

\begin{table}[t!]
  \centering
  {\small
    \begin{tabular}{l|l}
      Tenant & Characteristics\\
      \hline \hline
      BI & I/O-intensive SQL queries\\
      \hline DEV & Mixture of different types of jobs\\
      \hline APP & Small, lightweight jobs\\
      \hline STR & Hadoop streaming jobs\\
      \hline MV & Long-running, CPU-intensive\\
      \hline ETL & I/O-intensive, periodic but bursty\\
    \end{tabular}
    \vspace{-2mm}
    \caption{Tenant characteristics at Company ABC.}
    \vspace{-4mm}
  }
  \label{tbl:rf_users}
\end{table}

The BI and ETL users correspond directly to similarly-named
tenants. Among the other tenants, MV corresponds to the creation of
{\em Materialized Views} such as joined results of multiple tables as
well as statistical models created from the incoming data brought
through ETL.  The BI users and Data Scientists usually write their
queries and analysis programs on these materialized views. The APP
tenant runs jobs from a specific high-priority production
application. The DEV and STR tenants mostly comprise queries and
analysis programs being run as part of application development by
engineers and Data Scientists.  At Company ABC:

\squishlist
\item Jobs from the ETL and MV tenants have deadlines because any
  delay in these jobs will affect the entire daily operations of the
  company.  We have seen multi-day delays caused by deadline misses
  for the ETL and MV tenants that had significant business impact.
\item About 30\% of high-priority jobs in APP miss deadlines.
\item While all tenants want as low job response time as possible for
  completion of their jobs, BI, DEV, and STR are treated as
  ``best-effort'' tenants in that the goal is to provide their jobs as
  low response time as possible subject to meeting the requirements of
  the ETL, MV, and APP tenants.
\squishend

\subsection{Concern B}
\label{sec:concern-b}
Predictability of completion time for recurring jobs is a key need in
most companies.  This demand stems from ease of resource planning and
scheduling for dependent jobs. At Company ABC:

\squishlist
\item The completion of one of the recurring jobs of the ETL tenant
  varies between 5 and 60 minutes.
\item The completion of one of the recurring jobs of the MV tenant
varies between 2 and 6 hours. 
\squishend

\noindent While we observed that this variance is caused partly by
variation in the input sizes of the jobs across runs, these sizes
exhibit strong temporal patterns. For example, the input sizes of the
recurring jobs in ETL vary across days within a week, but remain
stable across multiple weeks.

\subsection{Concern C}
\label{sec:concern-c}

Resources can be wasted in multi-tenant systems due to reasons such
as: (i) task preemption; (ii) suboptimal configuration of resource
limits; and (iii) jobs in poorly-written queries being killed by DBAs.
Figure~\ref{fig:wasted_util} illustrates the impact of preemption
based on two tenants, A and B.  Tenant A first launched some tasks and
used up all the resources, yielding 100\% resource
utilization. Suppose B submitted tasks just after the resources were
grabbed by A. Then, without preemption, B's tasks will have to wait
until A's tasks finish; which could cause B to miss its deadlines.

On the other hand, suppose a preemption timeout of 1 time unit is
configured for B.  Then, preemptions will take place at time $2$
killing the most recently launched tasks of Tenant A, and B will
acquire the freed resources immediately after.  However, A's tasks
will lose the unfinished work and then be restarted at time $3$. The
region marked I in Figure~\ref{fig:wasted_util} corresponds to the
resources taken by the killed tasks of A that were preempted. This
figure shows that even if the resource utilization between time $1$
and time $3$ remained $1.0$ (100\%), the {\em effective utilization},
which excludes region I, is only 80\%.

\begin{figure}[t!]
  \centering
  \includegraphics[scale=0.4]{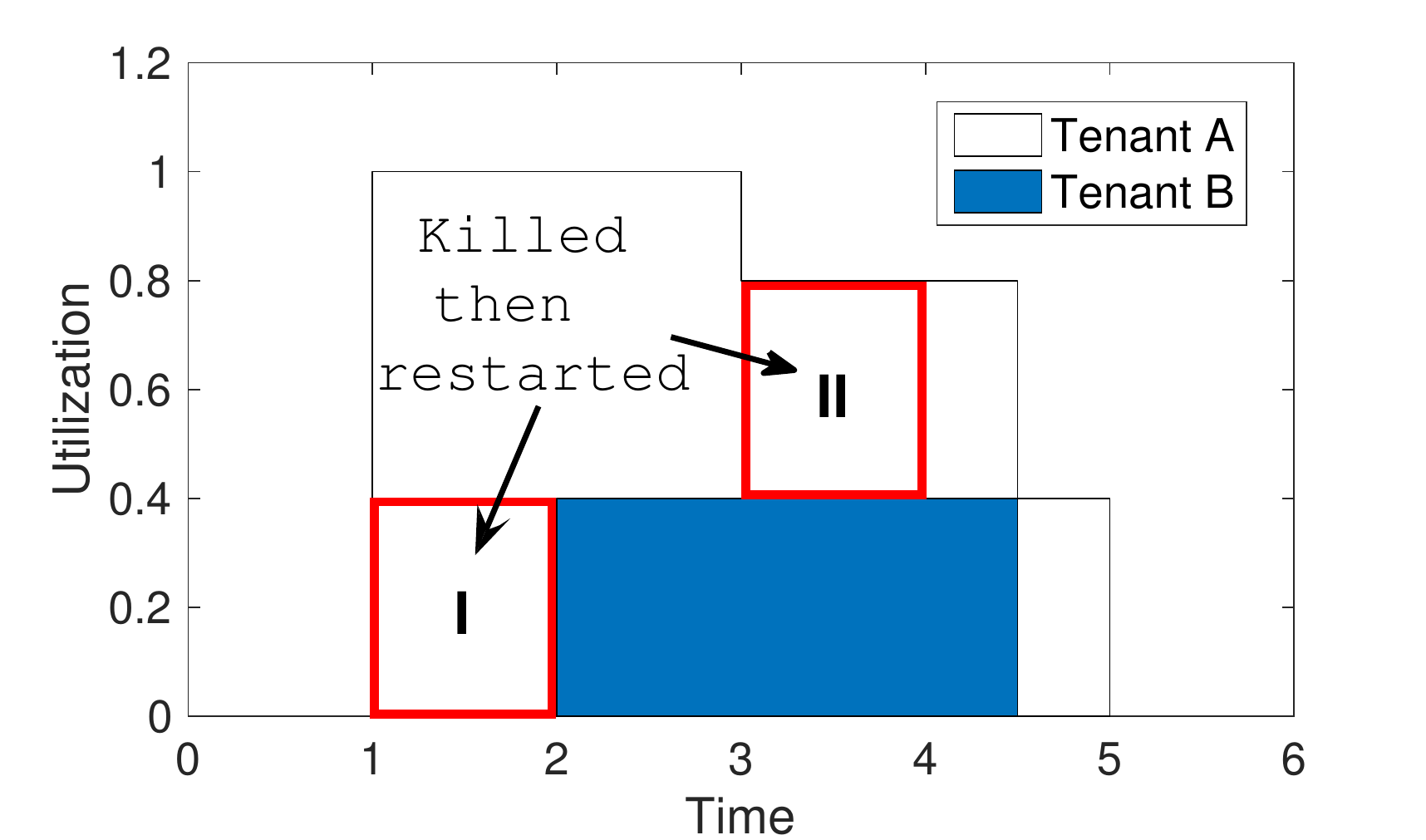}
  \vspace{-4mm}
  \caption{Wasted utilization due to preemption.}
  \vspace{-4mm}
  \label{fig:wasted_util}
\end{figure}

At Company ABC, 17.5\% of map tasks and 27.7\% of reduce tasks were
preempted for the jobs run by the MV tenant over a week interval.
This caused considerable amount of wasted resources, especially
because the reduce tasks of the MV tenant have long execution times.

\subsection{Concern D}
\label{sec:concern-d}
A resource allocation which meets the SLOs perfectly at one moment may
not be suboptimal at another moment due to various factors. First,
input data sizes for a tenant may vary considerably across shorter
time intervals while showing distinct patterns across longer
intervals. At Company ABC, ETL jobs process Web activity logs which
come in much smaller quantities on weekends.

\begin{figure}[ht]
  \centering
  \includegraphics[scale=0.35]{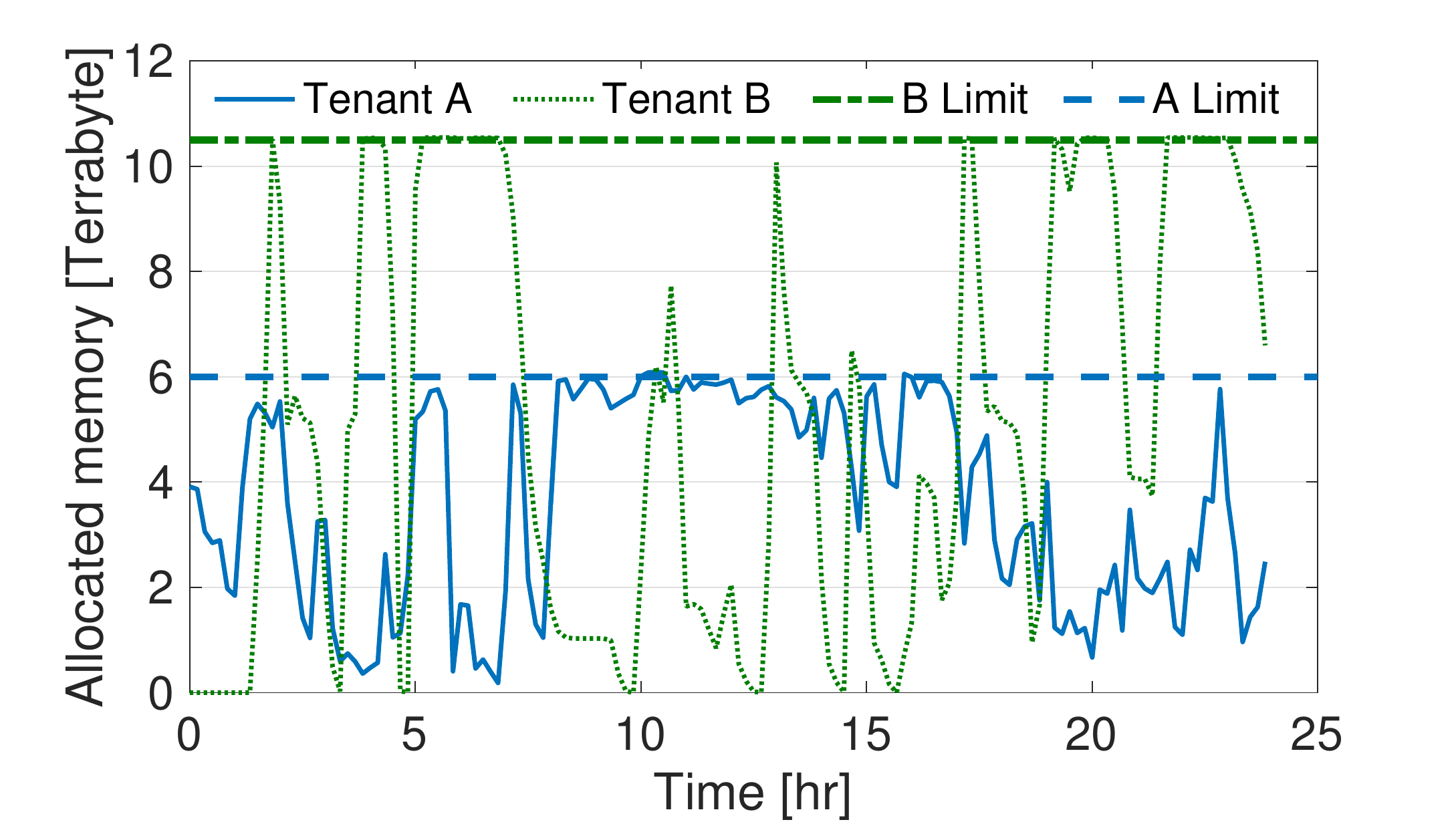}
  \vspace{-4mm}
  \caption{Memory consumption of two tenants during a day.}
  \vspace{-4mm}
  \label{fig:mem_use}
\end{figure}

Second, the resource demands of different tenants can be correlated
over time. For example, Figure \ref{fig:mem_use} shows the memory
consumption of two tenants at Company ABC over the course of a
day. The horizontal lines in the figure show the respective resource
limits that have been configured by the DBA to protect against
resource hoarding by tenants. Notice that while there are periods
where both tenants use up all available resources, there are other
periods where the configured resource limit prevents one tenant from
using the resources unused by the other.

\section{Overview of Problem}
\label{sec:problem}

From our interviews, two salient points emerged that summarize the
crux of what Tempo attempts to solve:

\squishlist
\item Workloads in multi-tenant parallel databases have SLOs. Current
  RMs do not provide easy ways to ensure that these SLOs are
  satisfied.
\item Current RMs require the DBA to estimate resources to meet the
  per-tenant SLOs, and then specify low-level RM configuration like
  resource shares, resource limits, and preemption timeouts in order
  to meet these SLOs.  This process is brittle and increasingly hard
  as workloads evolve, data and cluster sizes change, and new
  workloads are added.
\squishend

\subsection{SLOs}
\label{sec:slos}

From our interviews with users and DBAs, we identified five major
classes of SLOs.  The first class specifies job deadlines. For
recurring jobs, the deadline is either the start of the next run or an
absolute time point like 5:00 AM. The second class specifies that job
response time must be less than a given threshold. Such SLOs are often
associated with ad-hoc jobs. The third class is about ensuring that
each tenant gets a fair allocation of resources. In particular, when
the database is under contention, the proportion of resources
allocated to each tenant must adhere to predetermined values.  This
SLO class prevents individual tenants from monopolizing the resources
intentionally or otherwise.  Fourth, the resource utilization or job
throughput must be above a threshold. This SLO class generally serves
the interest of DBAs to maximize the return on investment (ROI) in the
cluster.  A fifth type of SLO orders the other SLOs in terms of
priority. This special SLO mandates that SLOs with higher priorities
be considered first when not all SLOs can be met with the resources
available.

\subsection{RM Configuration Space}
\label{sec:resource_model}
In this section, we will describe the typical set of configuration
parameters supported by modern RMs on a per-tenant basis. As we will
describe in later sections, Tempo adaptively computes the settings for
these per-tenant RM configuration parameters in order to maximize SLO
compliance.

Parallel databases decompose queries and analysis programs to DAGs
(Directed Acyclic Graphs) of jobs that each consist of one or more
parallel tasks. CPU, Memory, and other resources are allocated to
these tasks. The resources allocated to any tenant can be captured in
a fine-grained manner based on the start time, end time, and the
resource allocation vector $d$ for each of the tasks run on behalf of
the tenant.

In this paper, for ease of exposition, we will consider a
uni-dimensional representation of $d$ as an integer number of {\em
  containers} (or {\em slots}) as done in RMs like Mesos and YARN.
Namely, a task is run in a container that is allocated on behalf of a
tenant who submits the task. No two tasks can share the same
container.  The RM of a multi-tenant database system has a fixed total
number of containers that it can allocate across all tenants at any
point of time. This allocation is governed by a set of configuration
parameters for each tenant.  These parameters fall into three
categories, described next.

\vspace{1mm}
\noindent{\bf Resource Shares:} The resource share for a tenant
specifies the proportion of total resources that this tenant should
get with respect to other tenants. For example, suppose there are
three tenants A, B, and C with shares in the ratio 1:2:3
respectively. Suppose the database system has 12 containers that it
can allocate at any point of time. Then, if all tenants have tasks to
run, then tenants A, B, and C will get 2, 4, and 6 containers
respectively.

Suppose a tenant does not have tasks to run in its full quota of
resources. Then, the unused quota of resources will be allocated to
other tenants who have tasks to run.  This allocation will be
proportional to the resource shares of the other tenants. In the
example above, suppose tenant C has no tasks to run, but A and B have
many tasks to run.  Then, tenants A and B will get 4 and 8 containers
respectively.

\vspace{1mm}
\noindent{\bf Resource Limits:}
For any tenant, minimum and maximum limits can be specified for the
resources that this tenant can get at any point of time. In the
example above where tenants A, B, and C have shares in the ratio 1:2:3
respectively, suppose all tenants have many tasks to run, but the
maximum resource limit for tenant C is set to 3. Then, tenants A, B,
and C will get 3, 6, and 3 containers respectively. Limits are often
specified to ensure two things: (i) no tenant can monopolise all
resources, and (ii) critical workloads from a tenant can start running
as quickly as possible.

\vspace{1mm}
\noindent{\bf Resource Preemption:}
For any tenant, a configuration can be set to preempt---after a
certain time interval that the tenant should wait for---tasks from
other tenants that using resources that are rightly owed to this
tenant. Such preemption will free up resources for this tenant. There
are two levels of preemption timeouts. One level of preemption is when
the tenant's current resource allocation is below its configured
resource share. The other, and more critical level, is when the
tenant's current resource allocation is below its configured minimum
resource limit.

Preemption is important in multi-tenant systems.  Without preemption,
a low-priority tenant who submitted tasks earlier than a high-priority
tenant can cause the high-priority tenant to miss deadlines.
Preemption can be implemented by suspending tasks or by killing tasks
running in the container.  While task suspension is the preferred
mechanism, it is not supported in most multi-tenant systems that are
commonly used today. As we showed in Section \ref{sec:concern-c}, if
the two levels of preemption timeouts are not configured carefully,
then preemption by killing tasks can cause a lot of wasted work and
low resource utilization.

\subsection{Role of Tempo}
Our interviews revealed that DBAs manually tune the per-tenant RM
configuration parameters in order to meet tenant SLOs. For example, at
Company ABC, the RM configuration is tuned whenever tenants complain
about deadline or job response time SLOs not being met.  This process
is brittle because it is hard for the DBAs to take into account the
workload patterns and evolution, constant addition of new workloads,
and the multiple objectives and tradeoffs. The goal of Tempo is to
make this process easy and principled.

\section{Tempo}

As discussed in Section \ref{sec:intro}, Tempo is designed to bring
three properties to existing RMs: simplicity, self-tuning, and
robustness. As part of simplicity, Tempo introduces the concept of
{\em QS (Quantitative SLO)}.  A QS is a quantitative metric defined
per SLO to measure the satisfaction of the SLO at any point of time.
In Section \ref{sec:qs}, we will show how the QS concept supports
several tenant SLOs that arise in real-life use cases.

Operationally, the QS for an SLO can be thought of in two ways (recall
Section \ref{sec:resource_model}):
\squishlist
\item [1.] As a function $f\left(\pmb{x};w\right)$ where $w$ denotes
  the workload and $\pmb{x}$ is the vector representation of the
  parameters in the RM configuration used to allocate resources to
  process $w$.
\item [2.] As a function of the actual task resource allocation
  schedule (henceforth called {\em task schedule}) that is produced
  when the workload $w$ runs under $\pmb{x}$.
\squishend

\noindent As we will show in Section \ref{sec:qs}, it is conceptually
easier for humans to understand and use the QS concept when defined in
terms of the task schedule.  At the same time, Tempo needs the
$f\left(\pmb{x};w\right)$ notion in order to create a modular
architecture that provides self-tuning and
robustness. Figure~\ref{fig:framework} shows how this modular
architecture drives the repeated execution of Tempo's control loop.

\begin{figure}[ht]
  \centering \includegraphics[scale=0.45]{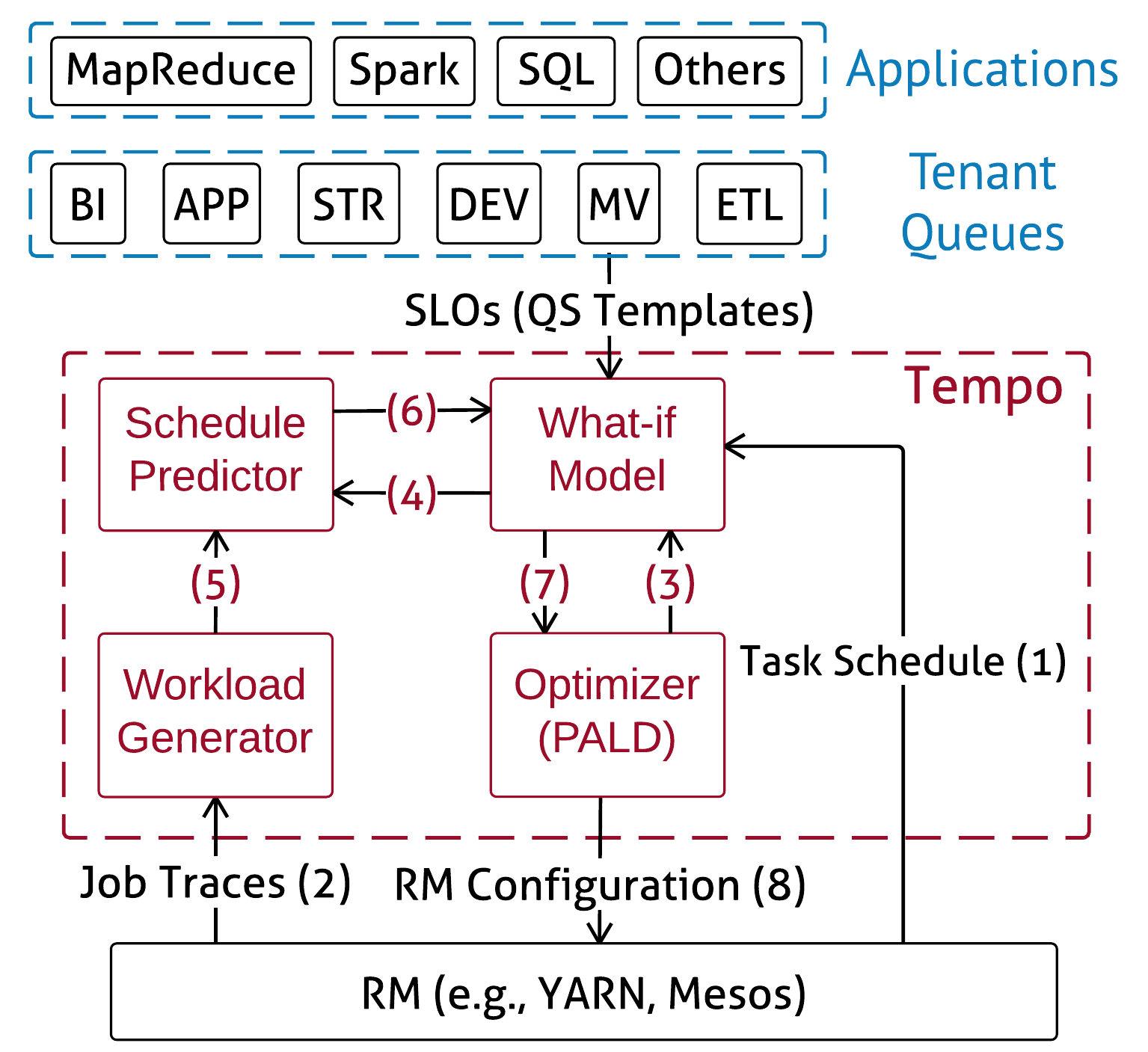}
  \vspace{-2mm}
  \caption{Tempo architecture: tenants specify SLOs using the QS
    templates, and Steps (1)-(8) form the Tempo control loop.}
  \label{fig:framework}
\end{figure}

The Tempo control loop consists of the eight steps denoted (1)-(8) in
Figure~\ref{fig:framework}. The inputs to the Tempo control loop are
the SLOs defined for each tenant (which can be specified conveniently
via predefined templates as discussed in Section \ref{sec:qs}). Step
(1) of the control loop extracts the recent task schedule for
evaluating QS metrics for the input SLOs under the current RM
configuration $\pmb{x}$.  Through Steps (2)-(8), Tempo replaces the
current RM configuration $\pmb{x}$ with a new one $\pmb{x}\prime$;
concluding one iteration of the control loop. Once the QS metrics for
the input SLOs under $\pmb{x}\prime$ are observed at Step (1) of the
next iteration, the Tempo control loop will revert the RM
configuration $\pmb{x}\prime$ back to $\pmb{x}$ if the currently
observed QS metrics do not dominate the previously observed ones. This
mechanism adds robustness in Tempo by guarding against performance
degradation during the self-tuning approach.

Steps (2)-(8) are orchestrated by Tempo's {\em Optimizer} which
applies a self-tuning algorithm called {\em PALD}. PALD is a novel
multi-objective optimization algorithm that we developed for the noisy
environments seen in production multi-tenant parallel database
systems.  As we will show in Section \ref{sec:theory}, PALD {\em
  provably} converges to a RM configuration that provides a
Pareto-optimal setting for the QS metrics of the input SLOs. In
addition, whenever available resources are insufficient to fully
satisfy all SLOs, PALD handles the SLO tradeoffs gracefully by
minimizing the largest regret across all SLO satisfactions as measured
by the QS metrics.

In Steps (2)-(8), the Optimizer explores a set of RM configurations by
proposing the RM configurations (3)-(4), getting the simulated task
schedule (6) of the workloads (5) based on the job traces (2). The
predicted QS metrics under these RM configurations are passed back to
the Optimizer (7) to compute a Pareto-improving RM configuration (8).
To implement these steps, the Optimizer uses three other components as
shown in Figure~\ref{fig:framework}: {\em Workload Generator}, {\em
  Schedule Predictor}, and {\em What-if Model}.

The Workload Generator replays historical job traces or synthesizes
workloads with given characteristics. The Schedule Predictor produces
the simulated task schedule of the generated workloads under given RM
configurations. The What-if Model evaluates the QS metrics for the
input SLOs using the simulated task schedule.  Together, the three
components enable the Optimizer to explore the impact of different RM
configurations on the input SLOs and use the PALD algorithm (described
in Section \ref{sec:theory}) to produce Pareto-optimal RM
configurations for these SLOs.

While proposing RM configurations in Step (3), the Optimizer
meticulously generates configurations only within a given maximum
distance to the currently used RM configuration. Tempo uses normalized
$l^2$-norm as the distance metric, and allows the DBA to specify the
maximum distance based on her risk tolerance.  This technique further
reduces the risk of causing dramatic impact on the running workloads
when applying a new RM configuration; which is particularly desirable
in production environments.

\section{QS: Quantifiable Metrics to Measure SLO Satisfaction}
\label{sec:qs}
A key design goal in Tempo was to provide a quantitative understanding
of how the workload and RM configuration impact each SLO. We developed
the QS metric which can be used to compare the relative SLO
satisfactions under different workloads and RM configurations.
Minimizing the QS metric improves the corresponding SLO. QS metrics
were motivated by the idea of loss functions in machine learning.

The QS metric for an SLO is defined as a function of the resulting
task schedule for a workload under a given RM configuration. Recall
from Section \ref{sec:resource_model} that a task schedule consists of
start time, end time, and the resource allocation $d$ for each of the
tasks run on behalf of a tenant.  For ease of exposition, $d$ can be
considered as an integer number of containers as done in RMs like
Mesos and YARN.

\subsection{QS Metrics for Popular SLOs}
\label{sec:qs-metrics}

We will now describe QS metrics for the common classes of SLOs that we
came across in our interview (recall Section \ref{sec:slos}).  Note
that SLOs and corresponding QS metrics can be defined at different
levels such as at the level of a recurring job, at the level of the
entire workload of a tenant, at the level of the entire cluster,
etc. In this section, we will define QS metrics at the job level, but
the ideas generalize.  Consider a certain interval of time $L$. Let
$J_i$ denote the set of jobs from tenant $i$ which was submitted and
completed during this interval. Let $T_i$ be the set of tasks
associated with $J_i$. Based on this notation, we can define the
following QS metrics for the common SLOs.

\vspace{1mm}
\noindent{\bf Low average job response time:} The QS metric for job
response time SLO takes the average across all jobs executed by the
tenant, as given by \eqref{eq:uds_lat} where $t^s_j$ and $t^f_j$ are
the submission and finish time of job $j$,
respectively. $\left|J_i\right|$ represents the cardinality of the job
set $J_i$.
\begin{align} \text{QS}_{\text{AJR}}\left(J_i\right) =
\frac{1}{\left|J_i\right|} \sum_{j\in J_i} \left(t^f_j -
t^s_j\right) \label{eq:uds_lat}.
\end{align}

\vspace{1mm}
\noindent{\bf Deadlines:} The QS metric for deadline SLO can be
defined as the fraction of jobs of a tenant that missed their
deadline. Let $t^d_j$ be the deadline of the job $j$, the deadline QS
metric can be defined as
\begin{align} \text{QS}_{\text{DL}}\left(J_i\right) =
\frac{1}{\left|J_i\right|} \sum_{j\in J_i} \mathbb{I}\left(t^f_j >
\gamma\left(t^f_j - t^l_j\right) + t^d_j\right) \label{eq:uds_dl},
\end{align}
where $\mathbb{I}\left(\cdot\right)$ is the indicator function, and
$\gamma$ is a slack (tolerance) when identifying the deadline
violation. That is, a job $j$ is considered violating the deadline
$t^d_j$ only if its completion is later than the deadline by a factor
$\gamma$ in terms of the job duration $t^f_j - t^l_j$.  The slack
makes the QS metric less sensitive to system variability.

\vspace{1mm}
\noindent{\bf High resource utilization:}
The resource utilization can be calculated as the integral of the
fraction of overall resources allocated to the tenant over the time
interval. This utilization amounts to the area of the shaded region in
Figure~\ref{fig:uds_util}. We can use the {\em dominant resource
  usage} when multiple resource types are considered
\cite{Ghodsi13,Ghodsi11,Shue12}.  Note that the dominant resource
usage is represented by a ratio between zero and one. When there is
only a single resource type, we normalize the resource usage.
\begin{figure}[h]
  \centering
  \includegraphics[scale=0.45]{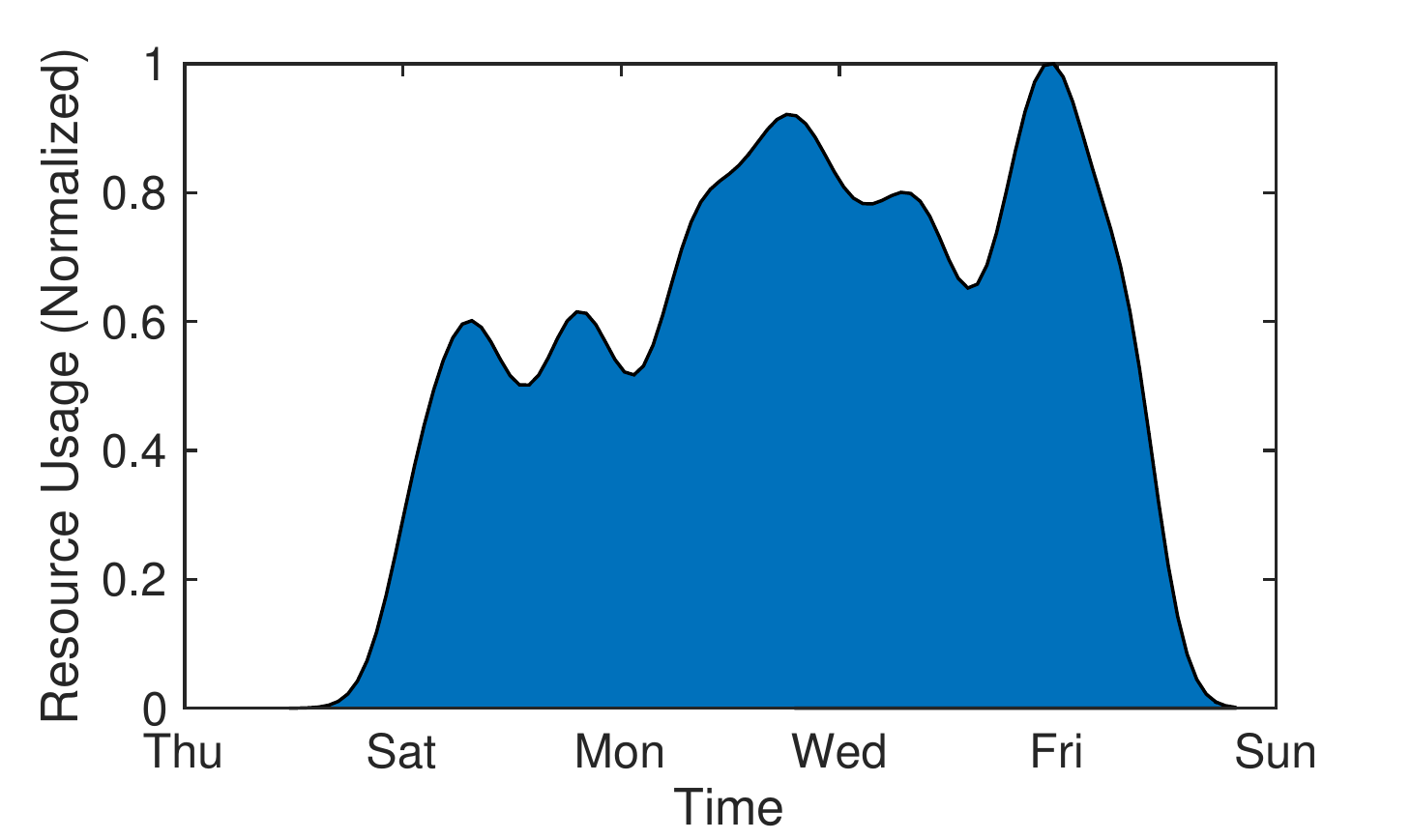}
  \vspace{-4mm}
  \caption{Normalized (dominant) resource usage over a period of
    time. The shaded area corresponds to the resources allocated to a
    particular tenant.}
  \vspace{-2mm}
  \label{fig:uds_util}
\end{figure}
Recall that the optimization minimizes the QS metrics. Thus, we can
define the QS metric for achieving high resource utilization as the
negative area in Figure~\ref{fig:uds_util}, which is given by
\begin{align}
  \text{QS}_{\text{UTIL}}\left(J_i\right) = -\frac{1}{L}\sum_{j\in
  T_i} d_j\left(t^f_t - t^l_t\right), \label{eq:uds_util}
\end{align}
where $L$ be the length of the interval, and $d_j$ is the amount of
resources allocated to task $j$. This QS metric can also be applied to
evaluate the impact of preemption, as illustrated in
Figure~\ref{fig:wasted_util}, by comparing the QS values computed
using all tasks versus using only tasks that were not preempted.

\vspace{1mm}
\noindent{\bf High job throughput:} The job throughput is defined as
the number of jobs submitted and completed within the interval. The QS
metric for achieving high job throughput is thus given by
\begin{align}
  \text{QS}_{\text{THR}}\left(J_i\right) = -\left|J_i\right| \label{eq:uds_thr}.
\end{align}

\vspace{1mm}
\noindent{\bf Resource fairness:} The fairness can be defined by
comparing the relative ratio of resource utilization used by the
tenants versus the desired ratio. This definition is also known as the
long-term fairness \cite{Tang14}. Let $c_i$ denote the desired share
of resources, the fairness QS metric follows
\begin{align*}
  \text{QS}_{\text{FAIR}}\left(J_i\right) = -\left|c_i + \text{QS}_{\text{UTIL}}\left(J_i\right)\right|.
\end{align*}

\subsection{QS Templates}

To further simplify the use of Tempo, we implemented {\em QS
  templates} to enable tenants to specify SLOs declaratively. A QS
template specifies: (a) a unique queue to which the tenant submits its
workload, (b) a predefined QS metric definition (e.g., the ones given
above), (c) one or more optional parameters associated with the
corresponding SLO (e.g., the value of a deadline, or a threshold on
job response time), and (d) an optional priority value (priorities are
incorporated by multiplying the QS metric with the priority value).
Tempo's QS templates make it easy to make the RM aware of SLOs such
as: ``Average job response time of tenant A must be less than two
minutes'', and ``No more than 5\% of tenant B's jobs can miss the
deadline of 1 hour.''

\section{Tempo's Theoretical Foundations}
\label{sec:theory}

\subsection{Multi-objective QS Optimization Problem}
\label{sec:multi-obj}

Tempo's Optimizer is given SLOs by the tenants and solves the
following problem:

\begin{align*}
  \text{minimize}
  \quad &\mathbb{E}\left[\left(f_1\left(\pmb{x};w\right),
          \cdots, f_k\left(\pmb{x};w\right)\right)\right] \label{eq:orig}
          \tag{SP1}\\ \text{subject to} \quad
        &\mathbb{E}\left[f_i\left(\pmb{x};w\right)\right] \leq r_i &
                                                                     \forall
                                                                     i=1,2,\cdots,k. \\ &\pmb{x} \in \mathcal{X}.
\end{align*}

Here,
$f_1\left(\pmb{x};w\right), f_2\left(\pmb{x};w\right), \cdots,
f_k\left(\pmb{x};w\right)$ denote QS functions for the $k$ SLOs from
tenants, under workload $w$ and RM configuration $\pmb{x}$. Since
measurements of the QS metrics will be noisy in a production database
system, the expectation $\mathbb{E}\left(\cdot\right)$ in
\eqref{eq:orig} is to average out the impact of noise. The vector
minimization in \eqref{eq:orig} is in Pareto optimal sense: an RM
configuration $\pmb{x}$ is said to dominate another configuration
$\pmb{x\prime}$, if for all $i\in\left[k\right]$,
$f_i\left(\pmb{x};w\right)\leq f_i\left(\pmb{x\prime};w\right)$, with
at least one inequality. An RM configuration is called weak
Pareto-optimal if no other RM configuration dominates it.
$\mathcal{X}$ is the RM configuration space defined in Section
\ref{sec:resource_model}. The $r_i$ values are used to specify
constraints that are part of SLOs.  Tempo converts all SLOs into such
constraints so that a solution to \eqref{eq:orig} will enforce all
SLOs.  In order to keep improving on the current RM configuration as
part of self-tuning, Tempo's control loop can use the QS value
attained for an SLO at the current configuration as the $r_i$ for the
next iteration.

As an example, let us consider two tenants {\em A} and {\em B}.  {\em
  A} has a deadline SLO while {\em B} is a best-effort tenant that
cares about getting the least possible job response times. For Tempo's
Optimizer, {\em A}'s deadline SLO will become a constraint
$\mathbb{E}\left[f_1\left(\pmb{x};w\right)\right] \leq r_1$, where
$r_1$ is {\em A}'s tolerable fraction of deadline violations and $f_1$
is the QS function for deadline SLO from Section \ref{sec:qs-metrics}.
{\em B}'s SLO for getting the lowest possible job response time can be
expressed as a constraint
$\mathbb{E}\left[f_2\left(\pmb{x};w\right)\right] \leq r_2$, where
$f_2$ is the QS function for response time SLO from Section
\ref{sec:qs-metrics} and $r_2$ is the average job response time for
{\em B}'s jobs obtained with the current RM configuration.

One can also prioritize certain SLOs over others in \eqref{eq:orig} by
weighting the corresponding QS functions.  For instance, to promote
the priority of an SLO whose QS is $f_i\left(\pmb{x};w\right)$, we can
replace the QS with $\alpha f_i\left(\pmb{x};w\right)$, where
$\alpha > 1$ is the magnitude of the promotion.

\subsection{Goals and Notation}

We now present a novel {\em PAreto Local Descent (PALD)} algorithm for
solving the multi-objective QS optimization problem \eqref{eq:orig}.
QS optimization poses two challenges which are not fully addressed by
existing methods: 1) the QS measurements are noisy due to inaccuracies
in job traces, choice of time intervals, etc.; and 2) QS metrics are
expensive to estimate as each prediction of a QS metric involves a
task scheduling simulation. Existing approaches roughly fall into
three classes. The first class uses evolutionary algorithms
\cite{Deb02,Deb06,Know06}, which are sensitive to noise and require
extensive QS predictions. The second class is prediction-based
\cite{Zulu13}, which inherently does not guarantee an optimal
solution. A third class achieves objectives which are different from
Pareto-optimality \cite{Mahd13,Kasi13}. In the following sections, we
describe PALD, and prove that it can handle the two challenges.

We denote vectors and matrices by boldface symbols. The simplified
notations $f_i$ and $f_i\left(\pmb{x}\right)$ are used interchangeably
to refer to the QS metric function $f_i\left(\pmb{x};w\right)$, and we
use $\pmb{f}\left(\pmb{x}\right)$ to refer to the vector of QS
functions. For each QS metric, we denote the average of $N$ measures
by $f_i\left(\pmb{x}\right)$.

The goal of PALD is to find a weak Pareto-optimal solution to
\eqref{eq:orig}. If a feasible solution exists, then the resulting RM
configuration satisfies the ``hard'' SLOs represented by the
constraints in \eqref{eq:orig}, while improving the ``best-effort''
SLOs. If there is no feasible solution, then the resulting RM
configuration balances the SLOs represented by the constraints based
on max-min fairness. This feature supports prioritizing the SLOs by
weighting the corresponding constraints.

\subsection{Proxy Model}
The key technique used in PALD is a proxy model, which transforms the
original problem \eqref{eq:orig} to a proxy problem \eqref{eq:proxy}
such that all solutions to the proxy problem are solutions to the
original one, but not the other way around. We show that the proxy
problem can be solved efficiently.

First, it should be noted that the well-known {\em weighted sum
  scalarization} (\cite{Boyd04})---which converts the
multi-dimensional QS vector to a scalar by taking a weighted sum of
the QS functions---does not apply in this case; for it does not ensure
the first set of constraints in the problem \eqref{eq:orig}. For
example, consider two RM configurations and two QS functions. Suppose
that the QS vectors corresponding to the two solutions are
$\left(5,5\right)^\top$ and $\left(0,7\right)^\top$, respectively. Let
$\pmb{r}=\left(6,6\right)^\top$. When the weights are equal, the
optimization using weighted sum scalarization yields the QS vector
$\left(0,7\right)^\top$, which does not dominate
$\pmb{r}=\left(6,6\right)^\top$.

Our solution PALD solves the following proxy problem:
\begin{align*}
  \text{minimize} \quad
  & \pmb{c}^{\top}\left[\pmb{f}\left(\pmb{x}\right)
    -\rho\max\left(\pmb{f}\left(\pmb{x}\right),\pmb{r}\right)\right]
    \label{eq:proxy}\tag{SP2}\\
  \text{subject to} \quad
  & \mathbb{E}\left[f_i\left(\pmb{x}\right)\right] \leq
    r_i &\forall i=1,2,\cdots,k. \\
  & \pmb{x} \in \mathcal{X}.
\end{align*}
Here, $\pmb{c}$, which is a positive vector, and $\rho < 1$ are two
parameters whose values will be described in Section
\ref{sec:params}. The parameter $\rho$ penalizes those QS functions
$f_{i}\left(\pmb{x}\right)>r_i$, and is independent of the vector
$\pmb{c}$. This is an advantage over {\em conic scalarization}
\cite{Kasi13}. One special case is that when $\rho=0$, the problem
\eqref{eq:proxy} becomes the weighted sum scalarization.

\begin{theorem}
  \label{thm:1} For any arbitrary positive vector $\pmb{c}$ and
parameter $\rho < 1$, every solution of \eqref{eq:proxy} is a solution
for \eqref{eq:orig}.
\end{theorem}
\begin{proof} Let $s_{i}\left(\pmb{x}\right) =
c_{i}\left[f_{i}\left(\pmb{x}\right) -
\rho\max\left(f_{i}\left(\pmb{x}\right), r_i\right)\right]$, the
objective of problem \eqref{eq:proxy} can be written as
  \begin{align} \sum_{i}s_{i}\left(\pmb{x}\right) & =
\sum_{i:f_{i}\left(\pmb{x}\right)\leq r_i}s_{i} \left(\pmb{x}\right) +
\sum_{j:f_{j}\left(\pmb{x}\right)>r_j}
s_{j}\left(\pmb{x}\right)\nonumber \\ & =
\sum_{i:f_{i}\left(\pmb{x}\right)\leq
r_i}c_{i}\left[f_{i}\left(\pmb{x}\right)-\rho
r_i\right]+\label{eq:s1}\\ & \phantom{=}
\sum_{j:f_{j}\left(\pmb{x}\right)>r_j}c_{j}
\left(1-\rho\right)f_{j}\left(\pmb{x}\right)\label{eq:s2}.
  \end{align} Both \eqref{eq:s1} and \eqref{eq:s2} are strictly
monotonically increasing with respect to $f_i\left(\pmb{x}\right)$, so
is the objective \eqref{eq:proxy}. Consider a solution $\pmb{x}$ for
\eqref{eq:proxy}. Suppose that $\pmb{x}$ is not a weak Pareto-optimal
solution for \eqref{eq:orig}, then there exists another weak
Pareto-optimal solution $\pmb{x\prime}$ for the problem
\eqref{eq:orig} that dominates $\pmb{x}$. However, this contradicts
the hypothesis that $\pmb{x}$ is a solution for the problem
\eqref{eq:proxy}, due to the monotonicity.
\end{proof}

\subsubsection{Parameters}
\label{sec:params} We now derive the parameters $\pmb{c}$ and $\rho$
in the proxy model \eqref{eq:proxy}. PALD uses the stochastic gradient
descent for solving the proxy problem, and the gradients are estimated
using the well-known LOESS \cite{Clev88}. Let $s\left(\pmb{x}\right)$
denote the objective of the proxy problem \eqref{eq:proxy}, the update
for each iteration is given by
\begin{equation}
\pmb{x}^{new}=\pmb{x}^{old}-\alpha\pmb{\nabla_{x}}s\left(\pmb{x}\right),
  \label{eq:sgd}\tag{SGD}
\end{equation} where $\alpha$ is the step size. The parameters
$\pmb{c}$ and $\rho$ are chosen such that the above update does not
increase those QS functions $f_{i}\left(\pmb{x}\right)\geq r_i$. We
thereby obtain $\forall i:f_{i}\left(\pmb{x}\right)\geq r_i$,
$-\alpha\pmb{\nabla_{x}}^{\top}f_{i}\left(\pmb{x}\right)
\pmb{\nabla_{x}}s\left(\pmb{x}\right)\leq0$, or equivalently
$\pmb{\nabla_{x}}^{\top}f_{i}\left(\pmb{x}\right)
\pmb{\nabla_{x}}s\left(\pmb{x}\right)\geq0$.  The parameters are also
chosen to best improve those violated constraints
$f_{i}\left(\pmb{x}\right)\geq r_i$. Fixing $\pmb{c}$, $\rho$ is
obtained by solving
\begin{align} &\arg\max_{\rho} \min_{i:f_i\left(\pmb{x}\right)\geq
r_i}
\pmb{\nabla_{x}}^{\top}f_{i}\left(\pmb{x}\right)\pmb{\nabla_{x}}s\left(\pmb{x}\right) \label{lp:param}\tag{RHO}\\
\text{s.t.}  \quad
&\pmb{\nabla_{x}}^{\top}f_{i}\pmb{\nabla_{x}}s\left(\pmb{x}\right)\geq0,
\quad \forall i:f_{i}\left(\pmb{x}\right)\geq r_i \nonumber\\
&\pmb{c}\geq 0, \quad \rho < 1. \nonumber
\end{align} Note that the objective of the proxy model
\eqref{eq:proxy} is not differentiable at points $\left\{\pmb{x}\in
\mathcal{X}:\pmb{f}\left(\pmb{x}\right) = \pmb{r}\right\}$, and we
need to condition on the subgradients. Let us first assume that
$\left.\partial s\left(\pmb{x}\right)/\partial
f_{j}\left(\pmb{x}\right)\right|_{x=r}=c_{j}\left(1-\rho\right)$.  The
objective of the problem \eqref{lp:param} can be rewritten as
\begin{align}
\sum_{j}c_{j}\pmb{\nabla_{x}}^{\top}f_{i}\pmb{\nabla_{x}}f_{j} -
\rho\sum_{j:f_{j}\left(\pmb{x}\right)\geq
r_j}c_{j}\pmb{\nabla_{x}}^{\top}f_{i}\pmb{\nabla_{x}}f_{j}. \label{eq:sign}
\end{align}

Based on the range of the subgradient of $s\left(\pmb{x}\right)$, we
can bound $\rho$. To satisfy the first set of constraints in the
problem \eqref{lp:param} at an indifferentiable point, we have that
\begin{equation} \min_{i,\frac{\partial s}{\partial
f_{j}}:f_{i}\left(\pmb{x}\right)\geq
r_i}\sum_{j}\pmb{\nabla_{x}}^{\top}f_{i}\pmb{\nabla_{x}}f_{j}\frac{\partial
s}{\partial f_{j}}\geq0\label{eq:rho}.
\end{equation} Now consider separately two cases $\rho \geq0$ and
$\rho <0$. When $\rho \geq 0$ the inequality \eqref{eq:rho} is
equivalent to $\forall i:\pmb{\nabla_{x}}f_{i}\neq \pmb{0} \land
f_i\left(\pmb{x}\right)\geq r_i$ that
\begin{align*}
\left(1-\rho\right)c_i\pmb{\nabla_{x}}^{\top}f_{i}\pmb{\nabla_{x}}f_{i}
& \geq -\sum_{j:j\neq i} \min_{\partial s/\partial f_{j}}
\pmb{\nabla_{x}}^{\top}f_{i}\pmb{\nabla_{x}}f_{j}\frac{\partial
s}{\partial f_{j}}\\ & = -\left(1-\rho\right)\sum_{\substack{j:j\neq
i,\\\pmb{\nabla_{x}}^{\top}f_{i}\pmb{\nabla_{x}}f_{j}\geq0}}c_j\pmb{\nabla_{x}}^{\top}f_{i}\pmb{\nabla_{x}}f_{j}\\
& \phantom{=} -\sum_{\substack{j:j\neq
i,\\\pmb{\nabla_{x}}^{\top}f_{i}\pmb{\nabla_{x}}f_{j}<0}}c_j\pmb{\nabla_{x}}^{\top}f_{i}\pmb{\nabla_{x}}f_{j},
  \end{align*} which simplifies to
\begin{align*}
0\leq\rho\leq\min_{\substack{i:\pmb{\nabla_{x}}f_{i}\neq
\pmb{0},\\f_i\left(\pmb{x}\right)\geq r_i}}
\frac{\sum_{j}c_j\pmb{\nabla_{x}}^{\top}f_{i}\pmb{\nabla_{x}}f_{j}}{\sum_{j:\pmb{\nabla_{x}}^{\top}f_{i}\pmb{\nabla_{x}}f_{j}\geq0}c_j\pmb{\nabla_{x}}^{\top}f_{i}\pmb{\nabla_{x}}f_{j}}.
\end{align*} Similarly, we can obtain the bound for the case $\rho <
0$. It should be noted that these bounds are useful only when the
following conditions are satisfied:
\begin{align}
\sum_{j}c_j\pmb{\nabla_{x}}^{\top}f_{i}\pmb{\nabla_{x}}f_{j}\geq 0,
\quad \forall i:\pmb{\nabla_{x}}f_{j}\neq\pmb{0}\land
f_i\left(\pmb{x}\right)\geq r_i.
  \label{eq:cond}
\end{align} These conditions can be satisfied for convex QS functions,
using the vector $\pmb{c}$ described in MGDA \cite{mgda12}. Combining
the results arrives at the optimal choice of $\rho$ for the problem
\eqref{lp:param}:
\[ \rho_{*}=\begin{cases} \min\limits_{\substack{
i:\pmb{\nabla_{x}}f_{j}\neq\pmb{0},\\ f_i\left(\pmb{x}\right)\geq r_i
}}\frac{\sum_{j}c_{j}\pmb{\nabla_{x}}^{\top}f_{i}\pmb{\nabla_{x}}f_{j}}{\sum_{j:\pmb{\nabla_{x}}^{\top}f_{i}\pmb{\nabla_{x}}f_{j}\geq0}c_{j}\pmb{\nabla_{x}}^{\top}f_{i}\pmb{\nabla_{x}}f_{j}},
& \rho\geq0\\
\max\limits_{\substack{i:\pmb{\nabla_{x}}f_{j}\neq\pmb{0},\\
f_i\left(\pmb{x}\right)\geq r_i
}}\frac{\sum_{j}c_{j}\pmb{\nabla_{x}}^{\top}f_{i}\pmb{\nabla_{x}}f_{j}}{\sum_{j:\pmb{\nabla_{x}}^{\top}f_{i}\pmb{\nabla_{x}}f_{j}<0}c_{j}\pmb{\nabla_{x}}^{\top}f_{i}\pmb{\nabla_{x}}f_{j}},
& \rho<0.
\end{cases}
\] The sign of the parameter $\rho$ depends on the last term of the
objective \eqref{eq:sign} as to maximize the objective. To deliver the
above optimal $\rho_*$, the vector $\pmb{c}$ must also satisfy the
conditions \eqref{eq:cond}.

To achieve max-min fairness of SLOs, PALD chooses $\pmb{c}$ that
improves the most violated constraint, through the following linear
program.
\begin{align*} \text{maximize} \quad &z\\ \text{subject to} \quad
&\pmb{J}_{i:f_i\left(\pmb{x}\right)\geq r_i} \pmb{J}^\top \pmb{c} \geq
z \pmb{1}\\ \quad & \pmb{c} \geq 0, \quad z \leq \epsilon.
\end{align*} Here, $\pmb{J}$ is the Jacobian of the QS vector, and
$\pmb{J}_{i:f_i\left(\pmb{x}\right)\geq r_i}$ denotes the rows of the
Jacobian $\pmb{J}$ indexed by $i:f_i\left(\pmb{x}\right)\geq r_i$.
$\epsilon$ is an arbitrary positive constant, and the solution vector
$\pmb{c}$ is normalized using any desirable metrics such as the
$l^2$-norm. The first set of constraints correspond to the QS
functions $f_i\left(\pmb{x}\right)\geq r_i$, and these are the only QS
functions that need to be convex in PALD. Thus, PALD provides better
support for non-convex QS optimization as compared to MGDA. Moreover,
randomly choosing different initial points can also help deal with
non-convex QS functions in this sense.

\section{What-if Model}
Tempo's Optimizer depends on the What-if Model to predict the values
of QS metrics for a given workload under a given RM configuration.
The What-if Model breaks each prediction into two steps and leverages
the Workload Generator and Schedule Predictor respectively for these
steps. Note that the QS metric is expressed as a function of the task
schedule of the workload under the given RM configuration.  The
Workload Generator is responsible for generating the workload, and the
Schedule Predictor is responsible for generating the task schedule
given the workload and the RM configuration.

\subsection{Workload Generation}
There are two ways to make workload information available to Tempo:
replaying historical traces or using a statistical model of the
workload (see Figure \ref{fig:framework}).  The statistical model,
which is usually trained from historical traces, has some key
advantages.  The model can be used to generate multiple synthetic
workloads with the same distribution in order to test the sensitivity
of parameter settings. More importantly, the model can be used to
generate synthetic workloads with extended characeristics such as a
growth in data size by 30\%. For example, we developed a statistical
model based on one month of historical traces from Company ABC's
production database workload. The workload distributions from Company
ABC (reported further in the evaluation section) are similar to the
distributions described in \cite{Ren12}. In particular, the task
duration approximately follows a lognormal distribution, and the job
arrival approximately follows a Poisson process.

\subsection{Fast Schedule Prediction}
\label{sec:sched_sim}
Given a workload generated as above, the Schedule Predictor in Figure
\ref{fig:framework} estimates the task schedule of the workload under
a given RM configuration.  Since the What-if Model needs to explore
the impact of many different RM configurations, fast prediction of
schedules can speed up Tempo's optimization process significantly.

For very fast task schedule simulation, we implemented a Schedule
Predictor for the RMs used in Hadoop, Spark, and YARN using time warp
mechanism \cite{jefferson82}. Our implementation computes the cluster
resource usage at only the submission time, tentative finish time, and
possible preemption time of each task, based on the workload
information and RM configuration parameter settings. This technique
helps the Predictor get rid of actually running the tasks as well as
synchronization within the RM.

To extend to other RMs, Tempo can leverage existing RM simulators that
have already been developed for several popular systems, such as Borg
\cite{abhi15}, Apollo \cite{Boutin14}, Omega \cite{Schw13}, MapReduce
\cite{Verma11,Hero11,Hammoud10}, and YARN \cite{SLS}.  Most of these
existing simulators are designed to reproduce the real-time behavior
of the RM, which is a superset of our goal of computing the task
schedule efficiently.

\section{Evaluation}
\label{sec:eval}

The evaluation consists of two parts. The first part validates the
components of Tempo based on real workload traces from the production
database at Company ABC running on Hadoop. The second part does an
end-to-end evaluation of Tempo using production workload traces from
Company ABC, Facebook, and multiple customers of Cloudera
\cite{Chen12}.  We apply Tempo to four real-life scenarios and show,
respectively, the improvements in job response time, resource
utilization, adaptivity to workloads variations, and predictive
resource provisioning.  In the experiments where a baseline
performance is needed for comparison, we used resource allocations as
determined by expert DBAs and cluster operators in Company ABC. The
following insights emerge from the evaluation:

\squishlist
\item Tempo can tailor the resource allocation to SLO-driven
  business-critical workloads, and offers tenants the freedom to
  specify SLOs.
\item Tempo improves the resource utilization by 15\%, and job
  response time for best-effort tenants by 50\% under 25\% slack
  without breaking the deadlines for production workloads.
\item Tempo effectively adapts to workload variations by periodically
  updating the RM configuration using a small window of most recent
  traces.
\item Tempo can further help DBAs and cluster operators determine the
  appropriate cluster size for their multi-tenant parallel database
  for the given SLOs and workloads, minimizing the overall resource
  costs.
\squishend

These results are due to: 1) informed resource allocation which takes
into account the workload characteristics revealed from historical
traces; and 2) optimized RM configurations aiming specifically for the
SLOs because Tempo makes the connection between the RM configuration
and SLOs more transparent and predictable.

\subsection{Validating the schedule prediction}
We begin by validating the task schedule prediction on a 700-node
production cluster at Company ABC. In particular, we measure the
accuracy of the prediction using one week's production workload from
six independent tenants, as described in Table~\ref{tbl:rf_users}. The
workload consist of approximately 60,000 jobs and 35 million
production tasks collected in a noisy environment where there were job
and task failures, jobs killed by users and DBAs, and
node blacklisting and restarts. Figure~\ref{fig:rf_keystat} shows the
key statistics of the workload.

\begin{figure*}
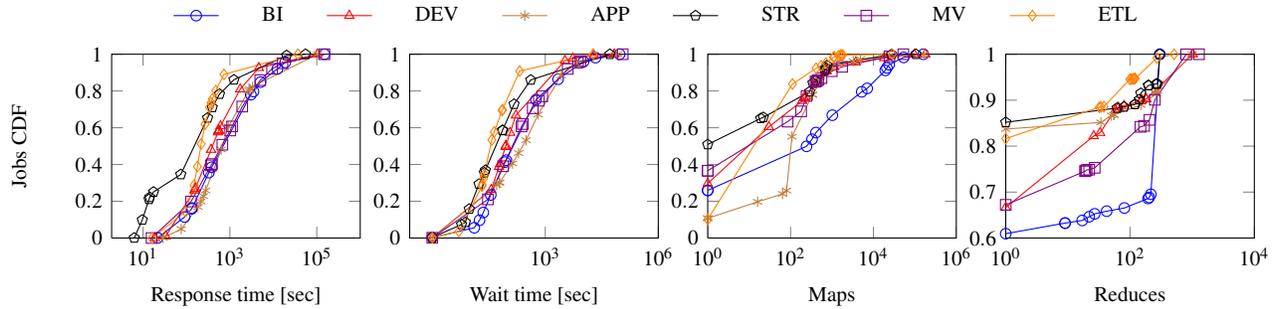

  \centering
  \input keystat.tex
  \vspace{-4mm}
  \caption{Key statistics of Company ABC's workloads.}
  \vspace{-2mm}
  \label{fig:rf_keystat}
\end{figure*}

The schedule prediction for the 35 million tasks from six tenants
takes just 4 minutes, or approximately 150,000 tasks per second. We
compare the predicted task schedule and the observed schedule based on
the traces, and compute the prediction error. Both the relative
absolute error (RAE) and the relative squared error (RSE) are used as
the error metrics. The RAE and RSE of tenant $i$ are defined
respectively as
\begin{align*}
  \begin{array}{cc}
    \text{RAE}_{i}= \frac{\sum_{j}\left|p_{ij}-l_{ij}\right|} {\sum_{j}\left|l_{ij}-\mathbb{E}_j\left[l_{ij}\right]\right|},
    &
      \text{RSE}_{i}=
      \sqrt{\frac{\sum_{j}\left(p_{ij}-l_{ij}\right)^{2}} {\sum_{j}\left(l_{ij}-\mathbb{E}_j\left[l_{ij}\right]\right)^{2}}}\end{array}.
\end{align*}
Here $p_{ij}$ and $l_{ij}$ represent the predicted and observed finish
time of job $j$ for tenant $i$, respectively. Table~\ref{tbl:pred_err}
gives the RAE and RSE for the estimated job finish time. As can be
seen, the highest error (24.4\%) was incurred for the MV tenant in
Company ABC. Most jobs from MV were long-running jobs, especially with
large duration of reduce tasks. We observed a considerable amount of
killed reduce tasks for MV due to preemptions. For killed and failed
tasks, the task start time and finish time are not recorded accurately
in workload traces; which explains why MV has a higher prediction
error than others.

\begin{table}
  \centering
  {\small
    \begin{tabular}{c c c|c c c}
      Tenant & RAE & RSE & Tenant & RAE & RSE\\ \hline \hline
      BI & 0.1585 & 0.2210 & STR & 0.1610 & 0.1463\\
      \hline DEV & 0.2195 & 0.2267 & MV & 0.2318 & 0.2437\\
      \hline APP & 0.1812 & 0.1599 & ETL & 0.1210 & 0.1908\\
    \end{tabular}
    \vspace{-2mm}
    \caption{Job finish time estimation errors for each tenant.}
    \vspace{-4mm}
  }
  \label{tbl:pred_err}
\end{table}

\subsection{End-to-end evaluation}
The end-to-end experiments involve four real-life scenarios, and were
performed on a 20-node Amazon EC2 cluster with m3.xlarge
instances. The production workload traces from Company ABC, Facebook,
and Cloudera customers were scaled and replayed on the EC2 cluster
using SWIM \cite{Chen12}. In addition, the initial RM configuration
was derived directly from the expert one created by DBAs for Company
ABC's production database.  As shown in the previous experiment, the
schedule prediction for processes 150,000 tasks per second. Each
end-to-end experiment involves approximately 30,000 tasks from two
tenants, and each Tempo control loop explores 5 RM configuration
candidates. Thus, one Tempo control loop requires prediction for
roughly 150,000 tasks, which takes one second.

\subsubsection{Mix of deadline-driven and best-effort workloads}
The first scenario involves two tenants running workloads which come
with the deadline SLO specified with $\text{QS}_{\text{DL}}$, and the
low average job response time SLO specified with
$\text{QS}_{\text{AJR}}$, respectively. The experiment aimed to obtain
an RM configuration which is better than the expert one used in
production. In particular, under the new RM configuration, {\em every}
job from the deadline-driven workload must complete no later than the
completion of the same job under the expert RM configuration. This is
a strict constraint, where the deadlines in $\text{QS}_{\text{DL}}$
are given by the completion times of deadline-driven jobs under the
expert RM configuration, and the corresponding $r_i = 0$ (for 0\%
deadline violations). Another constraint involving
$\text{QS}_{\text{AJR}}$ enforces that the average job response time
of the best-effort workloads under the new RM configuration cannot be
greater than the average job response time ($r_i$) under the expert
configuration.

When counting the number of deadline violations, a 25\% slack, i.e.,
$\gamma = 0.25$, is used in $\text{QS}_{\text{DL}}$ to reduce the
sensitivity to noise, since the workloads under the same RM
configuration with a slack 0 ($\gamma = 0$ in $\text{QS}_{\text{DL}}$)
can yield a large deadline violation fraction (up to 83\%).

\begin{figure}[h]
  \centering
  \includegraphics[scale=0.44]{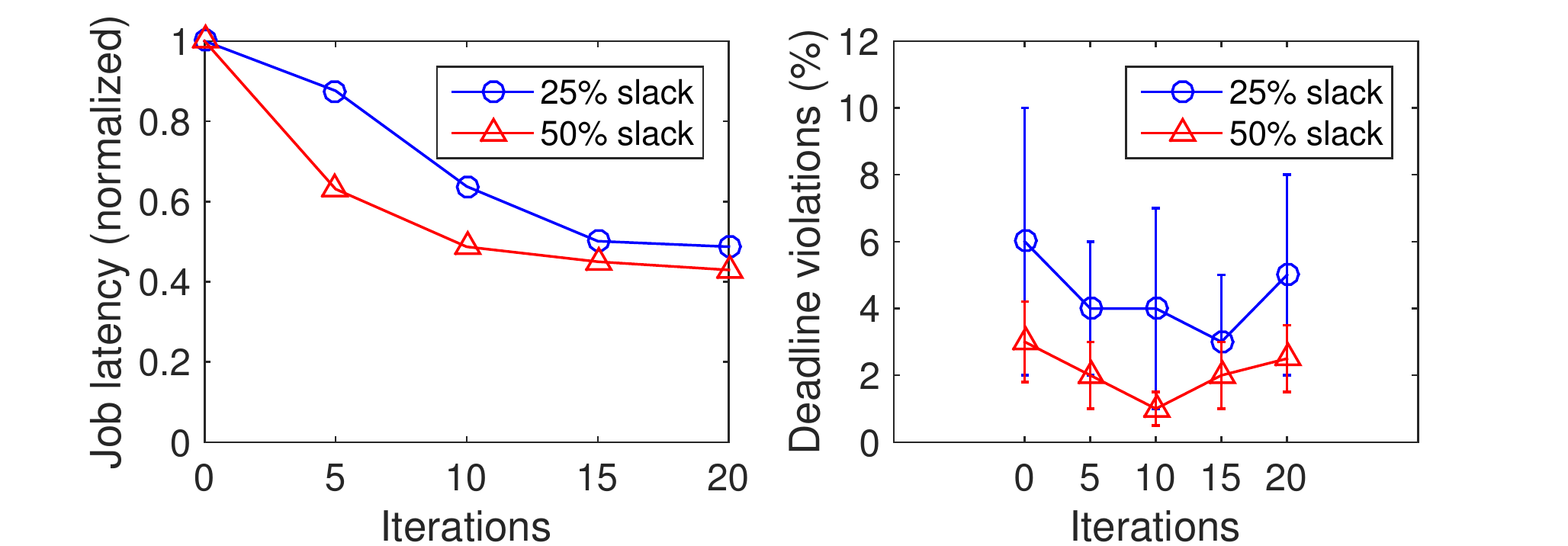}
  \vspace{-8mm}
  \caption{Average job response time for the best-effort tenant (left)
  and fraction of deadline violations for the deadline-driven tenant
  (right) at each iteration.}
\vspace{-2mm}
\label{fig:case1}
\end{figure}

Figure~\ref{fig:case1} shows the SLOs (QS values) at each iteration in
the Tempo control loop. At the iteration 0, the initial expert RM
configuration was used. The RM configuration was then iteratively
optimized by the Tempo control loop. It can be seen that, at
convergence, the improvements in average job response time of the
best-effort tenant are 50\% and 58\% for 25\% and 50\% slack,
respectively. The gap between the improvements is relatively small,
i.e., 8\%. One reason is that both improvements benefited from the
reduced contention for resources, which is confirmed in the next
experiment. In addition, the fraction of deadline violations first
drops and then breaks even at convergence. This trend is due to the
fact that once the Pareto frontier is reached, we cannot improve one
SLO without sacrificing another.

\subsubsection{Improving resource utilization}
In addition to the previous scenario, this experiment considered a
third SLO, high resource utilization, which is specified with
$\text{QS}_{\text{UTIL}}$. We focused exclusively on MapReduce
workloads due to the observation of significant task preemptions in
production. The experiment added two constraints corresponding to the
map container utilization and reduce container utilization,
respectively. The $r_i$'s were set according to the measured map and
reduce container utilization under the expert RM configuration. The
results show fewer preemptions under the Tempo optimized RM
configuration as well as improvements in job response time subject to
the deadline SLOs.

As we discussed, preemption happens when a tenant has been starved for
a certain period of time (the configured preemption timeout), killing
a certain number of most recently launched tasks from other
tenants. Thus, preemption results in lost work and decreased resource
utilization. Each tenant can specify a per-tenant preemption timeout
in the RM configuration, and these settings are difficult to get right
due to their complex connections to workloads and SLOs.

\begin{figure}[h]
  \centering \includegraphics[scale=0.43]{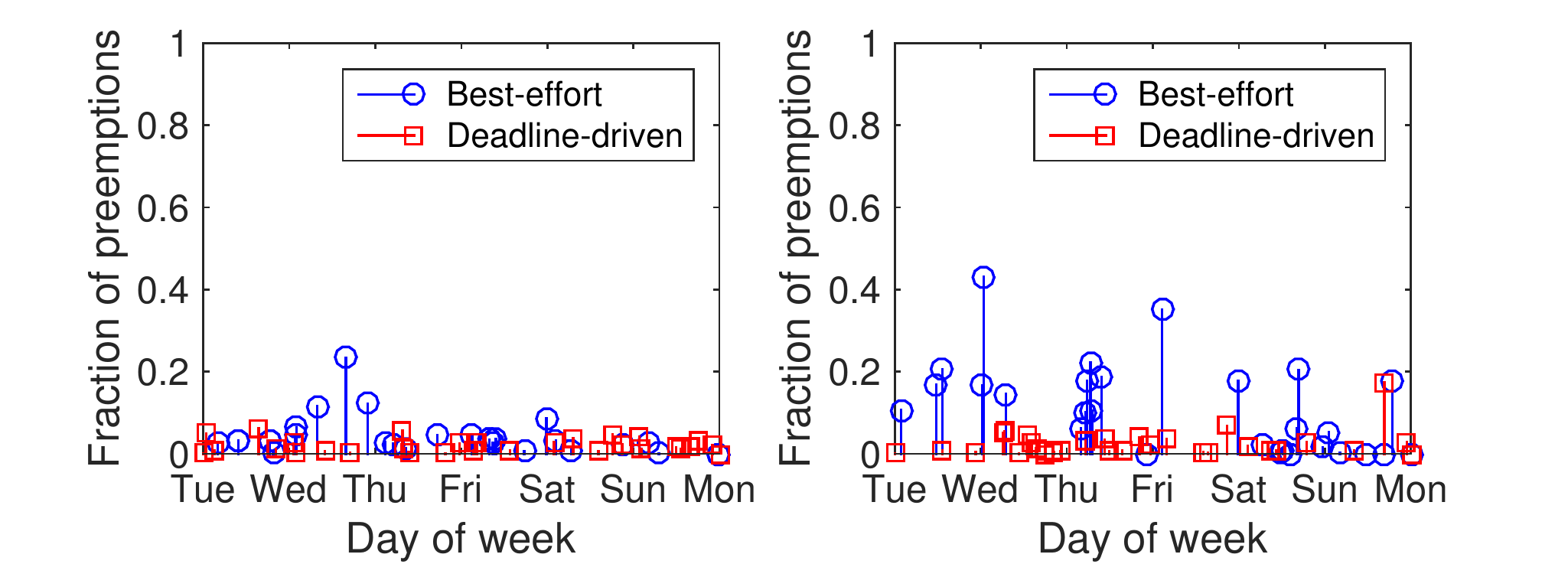}
  \vspace{-8mm}
  \caption{Task preemptions for MapReduce workloads at Company ABC. On
    the left shows the preempted map tasks, and the preempted reduce
    tasks are given on the right.}
  \vspace{-2mm}
  \label{fig:case2_preempt}
\end{figure}

We observed a significant number of preempted MapReduce tasks on the
production cluster at Company ABC. Figure~\ref{fig:case2_preempt}
shows the map and reduce preemptions over the period of one
week. During this period, 6\% map tasks and 23\% reduce tasks had been
preempted, and the reduce preemptions were mostly from the best-effort
tenant. The main reason was that the workloads of the best-effort
tenant contain mostly long-running reduce tasks, as shown in
Figure~\ref{fig:case2_taskdur}.

\begin{figure}[h]
  \centering \includegraphics[scale=0.44]{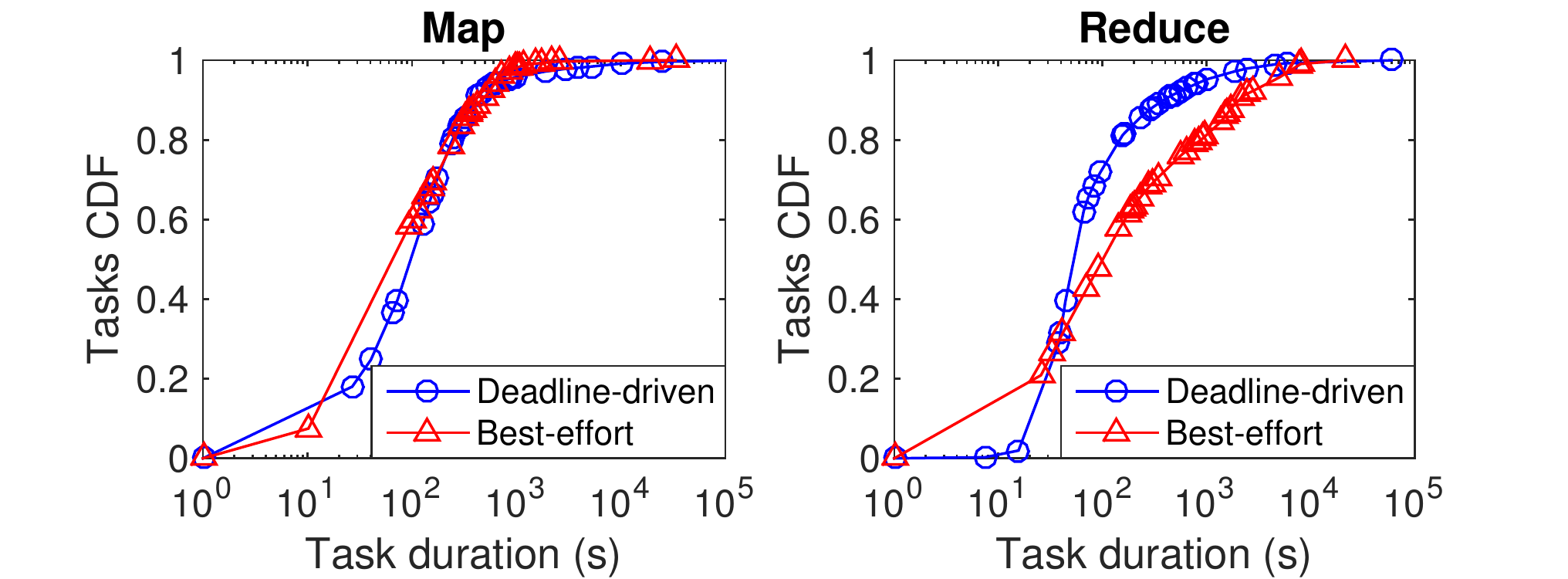}
  \vspace{-8mm}
  \caption{Task duration distributions for MapReduce workloads at
    Company ABC.}
  \vspace{-2mm}
  \label{fig:case2_taskdur}
\end{figure}

Figure~\ref{fig:case2_bar} shows the SLOs under the original expert RM
configuration and the Tempo optimized RM configuration. As can be
seen, the optimized resource allocation delivers 22\% improvement in
the average job response time of the best-effort tenant workloads and
10\% in the deadline QSs. Another improvement is in the utilization of
reduce containers, while the utilization of map containers remains at
the same level. The results are consistent with our observations of
preemption statistics, and the improvements in reduce container
utilization is due to the alleviated preemptions. In particular, the
preemption timeout settings in the Tempo-optimized RM configuration
had been self-tuned appropriately to the workload distribution.

\begin{figure}[h]
  \centering \includegraphics[scale=0.5]{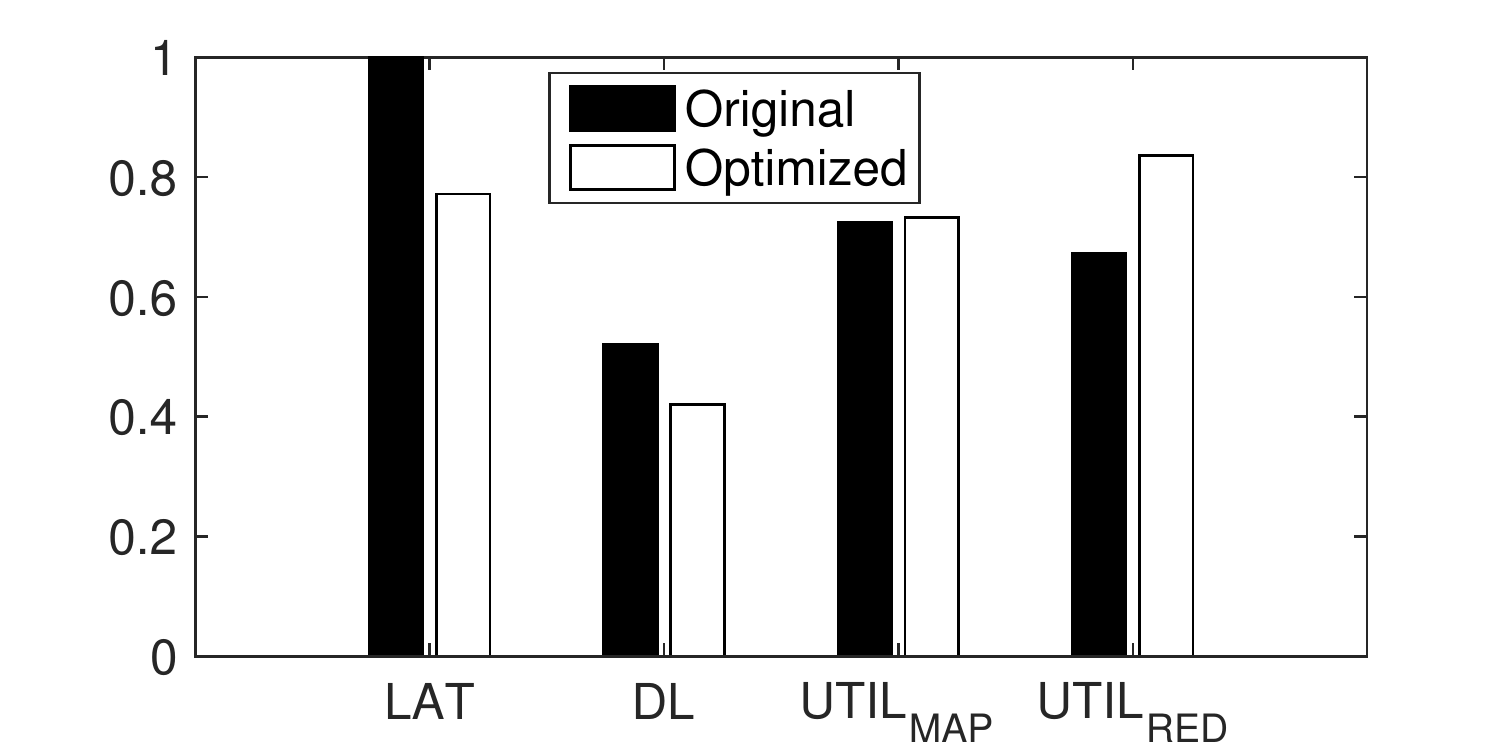}
  \vspace{-4mm}
  \caption{SLOs under the original and optimized (slack = 0) RM
    configuration: $\text{AJR}$, $\text{DL}$,
    $\text{UTIL}_{\text{MAP}}$, and $\text{UTIL}_{\text{RED}}$ are the
    average job response time of the best-effort tenant, fraction of
    deadline violations for the deadline-driven tenant, map container
    utilization, and reduce container utilization, respectively.}
  \vspace{-4mm}
  \label{fig:case2_bar}
\end{figure}

\subsubsection{Adaptivity to workload variations}
In this experiment, we applied Tempo to meet SLOs under slowly
changing workload distributions. Figure~\ref{fig:case3_instlat}
depicts the instant job response time distribution for deadline-driven
and best-effort tenants. The instant job response time is computed
using the moving average of a 30-min window. As can be seen, the
instant job response time of deadline-driven workloads exhibits a
periodic pattern while the job response time of the best-effort
workloads changes dramatically over time.

\begin{figure}[h]
  \centering \includegraphics[scale=0.44]{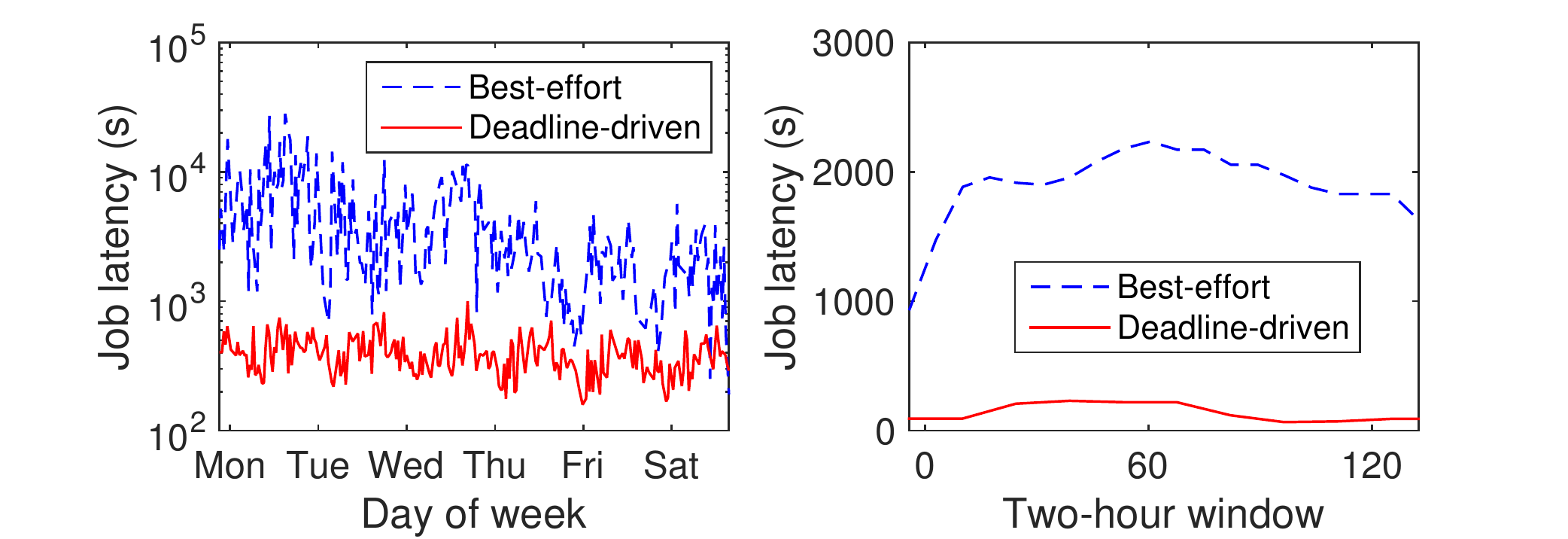}
  \vspace{-8mm}
  \caption{Instant job response time distributions. On the left shows
    the production workloads of Company ABC over the period of a
    week. On the right gives the two-hour experiment workloads on EC2
    using Facebook and Cloudera traces.}
  \vspace{-2mm}
  \label{fig:case3_instlat}
\end{figure}

Recall that each iteration of Tempo control loop uses a fixed-length
interval of most recent job traces as input. The next experiment
evaluates how different interval lengths impact Tempo's performance.

\begin{figure}[h]
  \centering \includegraphics[scale=0.5]{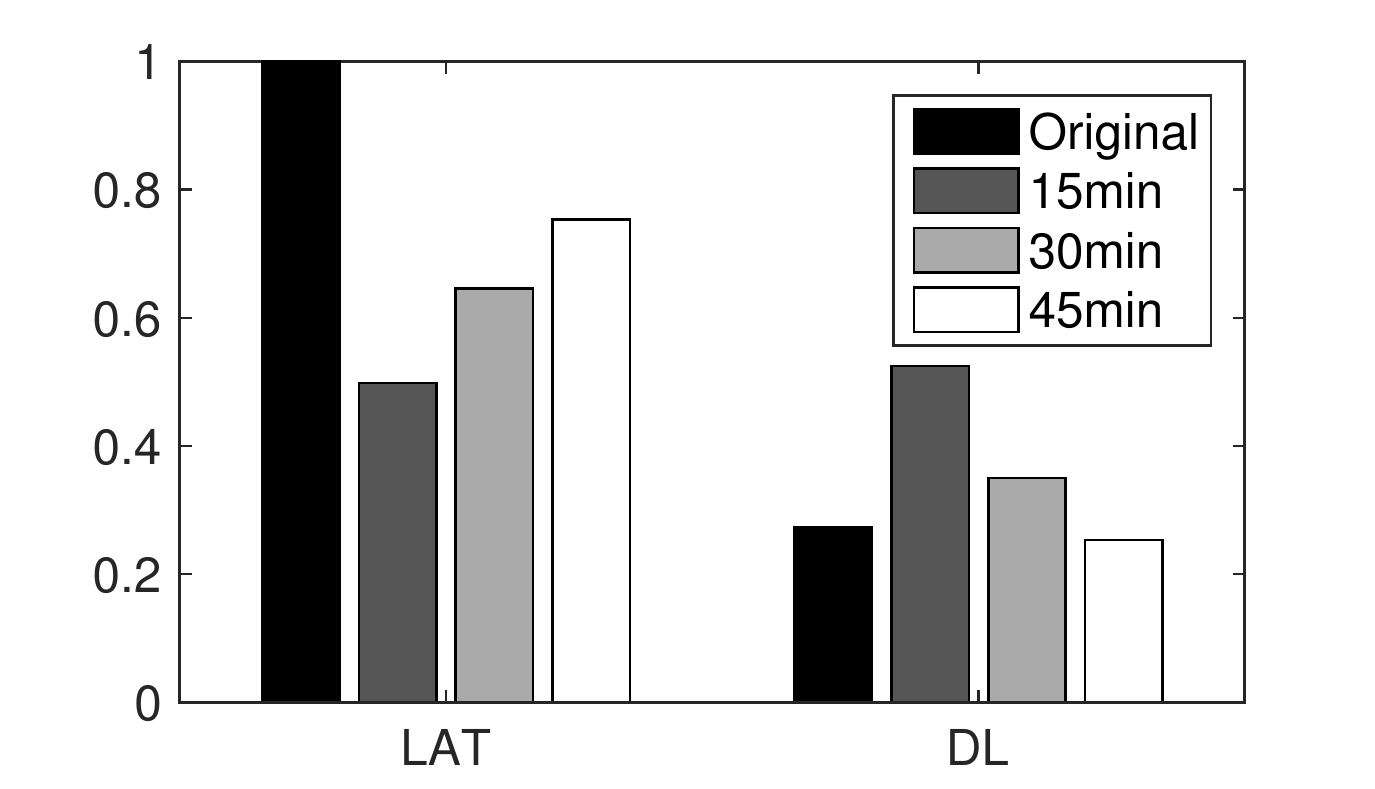}
  \vspace{-4mm}
  \caption{SLOs for different interval lengths in Tempo control
    loop. AJR denotes the normalized average job response time of the
    best-effort workloads, and DL represents the fraction of deadline
    violations (computed via $\text{QS}_{\text{DL}}$ with slack
    $\gamma=25\%$).}
  \vspace{-2mm}
  \label{fig:case3_res}
\end{figure}

Figure~\ref{fig:case3_res} shows the SLOs under the original expert RM
configuration and Tempo optimized RM configurations for interval
length 15min, 30min, and 45min. Similarly, the experiment uses SLOs
specified with $\text{QS}_{\text{AJR}}$, and $\text{QS}_{\text{DL}}$
(25\% slack). As can be seen, a small window size favors the average
job response time of the best-effort workloads while leading to a
higher percent of deadline violations. According to the results, the
45min interval length yields a similar fraction of deadline violations
as the original RM configuration, but a 22\% improvement in the
average job response time of the best-effort workloads. The results
show that Tempo can adapt to workload variations using a small
interval length.

\subsubsection{Resource provisioning and cutting costs}
The last experiment demonstrates the application of Tempo to resource
provisioning, estimating the minimum amount of resources needed to
meet the given SLOs.  This application can help users do better
resource planning and cut overprovisioning costs. In addition, this
application can bridge the gap in resource allocation between the
development cluster and the production cluster, that is, converting
the resource allocation on the development cluster for use in the
production cluster.

The experiment involves running the same given deadline-driven
workloads and best-effort workloads on three EC2 clusters with 20
nodes (100\%), 10 nodes (50\%), and 5 nodes (25\%),
respectively. Tempo was used to estimate the SLOs of the workloads
when executed on the 100\% cluster, using traces respectively from the
100\% cluster, 50\% cluster, and 25\% cluster. This experiment mimics
the scenario in which users collect traces of the workload on the
current cluster, and would like to know how a new cluster size will
impact the SLOs. (From our experience, this use case is common at
companies like LinkedIn and Yahoo.) In this case, Tempo can serve as a
key component in the decision-making for resource provisioning.

\begin{figure}[h]
  \centering \includegraphics[scale=0.43]{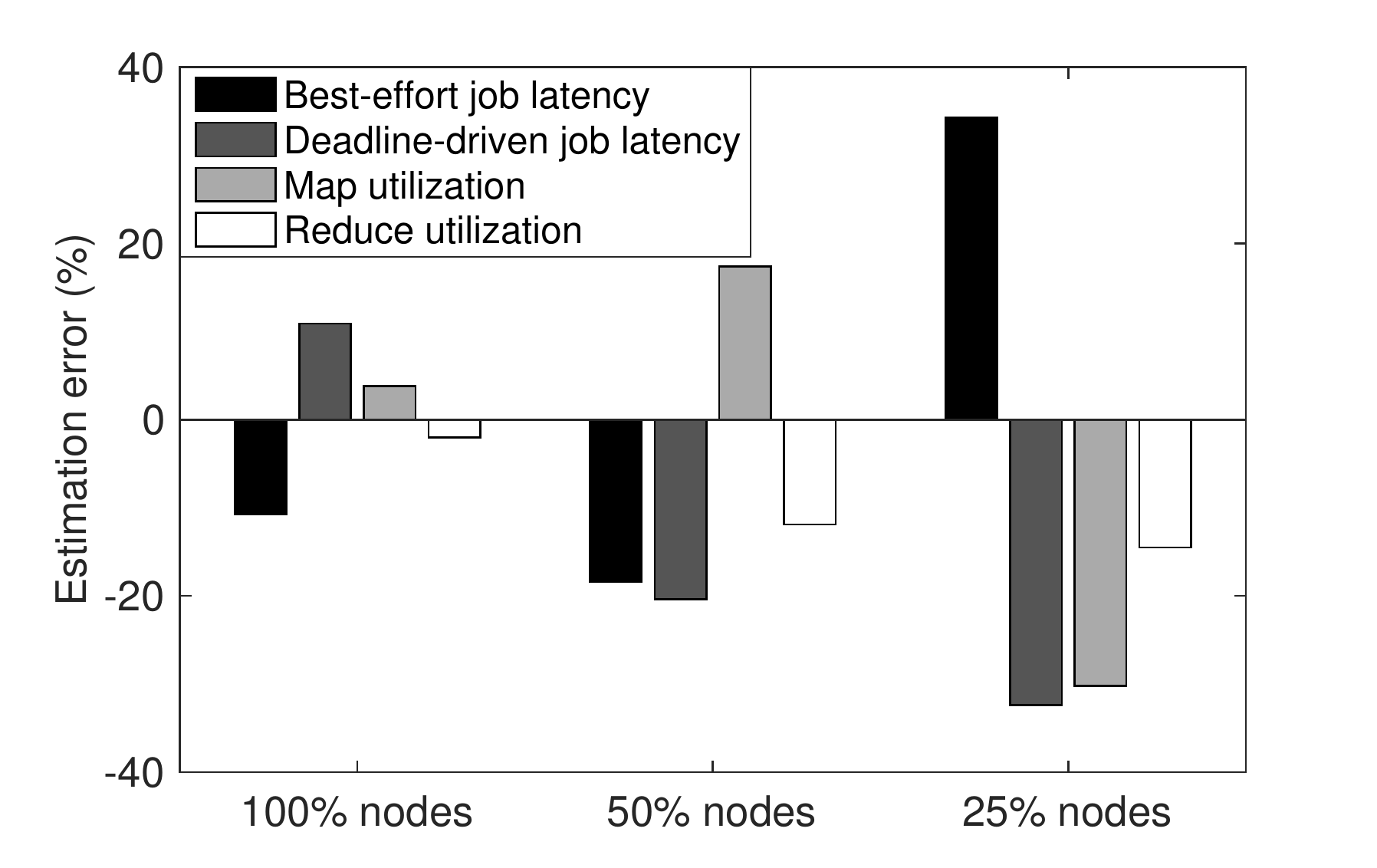}
  \vspace{-4mm}
  \caption{Errors in SLO estimation using traces based on equal and
    smaller cluster sizes.}
  \vspace{-2mm}
  \label{fig:case4}
\end{figure}

Figure~\ref{fig:case4} gives the SLO estimation errors using traces
from equal and smaller clusters. As can be seen, Tempo can
predict---with the error no more than 20\%---the SLOs of the current
workloads run on a double-size cluster; using traces collected from
the current cluster. Predicting the SLOs of the current workloads run
on a quadruple-size cluster results in a maximum error of 35\%.

\section{Related Work}
\label{sec:related}

Most Resource Managers (RMs) that are deployed on multi-tenant ``big
data'' database systems today like RedShift, Teradata, Vertica,
Hadoop, and Spark are based on simple resource allocation principles
such as static resource partitioning \cite{redshift}, {\em max-min
  fairness} \cite{FAIR,CAPACITY}, or {\em dominant resource fairness}
\cite{Ghodsi13,Ghodsi11,Shue12}.  Quincy \cite{Isard09} is a Fair
scheduler which takes into account data locality as a
preference. Choosy \cite{Ghodsi13} further extends the max-min
fairness to support hard job placement constraints.

Parallel database systems like IBM DB2 PE, RedShift, Teradata, and
Vertica have RMs (usually called Workload Managers) that allow DBAs to
specify {\em filter}, {\em throttle}, and {\em workload} rules to
dynamically adjust the resource allocation of tenants. For example, a
filter rule can reject unwanted logon and query requests before they
are executed; a throttle rule can restrict the number of requests
simultaneously executed against a database object. In addition, some
databases allow user-defined events relevant to workload management to
be defined and actions taken based on them.  RedShift and Vertica use
resource pools where each pool has parameters such as resource limits,
priorities, and maximum concurrency like the RM configuration
described in Section \ref{sec:resource_model}.

Mesos \cite{Hindman11} introduces two-level resource allocation to
support other custom RMs and improves data locality while allocating
resources to tasks.  YARN \cite{Vavil13} separates computation and
resource management by introducing a resource layer. Corona
\cite{Corona} uses a dedicated job tracker for each job to improve
cluster utilization, and adopts a push-based scheduling model to
improve the scalability.  Omega \cite{Schw13} improves the scalability
of RMs using parallelism, shared state, and lock-free optimistic
concurrency control. Sparrow \cite{Ous13} is a stateless scheduler
that aims for low job response time in task scheduling, leveraging
load-balancing techniques. It supports multiple schedulers as well as
job and task placement constraints.  Fuxi \cite{Zhang14} enhances the
fault tolerance and scalability of RMs by introducing user-transparent
failure recovery features and a failure detection mechanism. Apollo
\cite{Boutin14} is another shared-state scheduling framework, and
takes into account the data locality and server load to achieve
high-quality scheduling decisions. Apollo also introduces
opportunistic scheduling to improve cluster utilization, and
unexpected cluster dynamics by detecting abnormal runtime behaviors.

Unlike Tempo, all the above RMs leave it to DBAs and cluster operators
to create specific resource allocation policies as plugins.  Other RMs
have been proposed recently that take certain tenant-level SLOs into
consideration.  Rayon \cite{Curino14} introduces reservation-based
scheduling to make resource allocation more predictable by planning in
advance. The planning also increases the resource utilization. Rayon
considers two types of jobs, production and best-effort jobs, and
seeks to meet the completion deadlines of the production jobs and
reduce completion latency of the best-effort jobs. However, unlike our
work, Rayon is 1) intrusive in that it makes changes that cannot be
applied easily to a wide class of RMs in multi-tenant database
systems; 2) supports only two types of SLOs; 3) can be potentially
wasteful by reserving resources in the presence of job failures, and
4) may still need the user/DBA to determine how much resources need to
be reserved to meet SLOs, which is difficult.

In \cite{Kc10}, the authors develop a deadline estimation model and
apply real-time scheduling to meet job deadlines. ARIA \cite{ARIA}
provides support for job deadlines by profiling jobs and modeling
resource requirements in order to complete before the deadline. WOHA
\cite{Li14} improves workflow deadline satisfactions in Hadoop.
Tetris \cite{Grandl14} avoids resource fragmentation by introducing
multi-resource packing of tasks to machines.  Tetris improves the
average job completion time, and achieves high cluster makespan. The
authors show that task packing can also work without significantly
violating the fairness.  Tetris lacks the rich support for SLOs that
Tempo provides.

Pisces \cite{Shue12} focuses on datacenter-wide per-tenant performance
isolation and fairness for multi-tenant cloud storage.  Amoeba
\cite{Ana12} delivers lightweight elasticity to compute clusters by
splitting original tasks into smaller ones, and allowing safe exit of
a running task and later resuming the task by spawning a new task for
its remaining work. These features can help reduce the cost of
preemption.  Pulsar \cite{Angel14} is a resource management framework
similar to Tempo. Pulsar provides end-to-end performance isolation
through a {\em virtual datacenter abstraction (VDC)} which essentially
encapsulates resource demand forecasting and QS metrics for each
tenant. However, Pulsar focuses on sufficient resource scenarios and
does not support multiple performance goals associated with each
tenant as well as their trade-offs in a limited resource setting. The
effectiveness of Pulsar also relies on the accuracy of user-specified
cost functions in VDCs and resource demand estimation. Unlike Tempo,
robust resistance to noise (inaccuracies in both cost and demand
estimation) is not guaranteed in Pulsar.

Personalized Service Level Agreements (PLSAs) are introduced in
\cite{Ortiz15} which serve as cost models connecting resources and
SLOs. PLSAs can be used as QS metrics in Tempo based on SQL queries
executed by each tenant.

We also briefly summarize main related work on the various classes of
multi-objective optimization problems, and discuss their limitations
with respect to what Tempo provides.
\squishlist
\item \textbf{Convex and noiseless objective function:} This scenario
  can be solved using the well-known weighted sum scalarization. An
  alternative method, multiple gradient descent (MGDA), is proposed in
  \cite{mgda12} for solving the weighted sum scalarization
  problem. Another notable approach is Normal Boundary Intersection
  (NBI) \cite{Das98}, which allows evenly finding the Pareto frontier.
\item \textbf{Non-convex and noiseless objective function:} One recent
  approach is conic scalarization (CS) \cite{Kasi13}, in which a new
  scalarization method is proposed to achieve Benson and Henig proper
  efficiency. However, the choice of the weight vector is not
  addressed in CS. There are also several evolutionary algorithms
  \cite{Deb02,Deb06,Know06}, which assume noiseless or low-noise
  situations, but are expensive to run.
\item \textbf{Convex and noisy objective function:} In \cite{Mahd13},
  the authors present a stochastic primal-dual algorithm for a single
  objective but multiple constraint problem; and it does not seek a
  Pareto-optimal solution. Noticeably, MGDA \cite{mgda12} could
  potentially be extended to the noisy scenario, but it does not
  handle the first set of constraints in the problem \eqref{eq:orig}.
\item \textbf{Non-convex and noisy objective function:} A
  prediction-based approach is proposed in \cite{Zulu13}. This method
  estimates the Pareto frontier under the Gaussian assumption of the
  loss objectives, and does not specifically take into account the
  preference constraints.
\squishend

\section{Conclusion and Future Work}

Meeting the SLOs of business-critical workloads while achieving high
resource utilization in multi-tenant ``big data'' database systems is
an important problem.  The vast majority of resource
allocators/schedulers deployed on multi-tenant database systems today
rely on the DBAs to configure low-level resource settings. This
process is brittle and increasingly hard as workloads evolve, data and
cluster sizes change, and new workloads are added. In this paper, we
presented a framework, Tempo, which enables DBAs to work with
high-level SLOs conveniently. We demonstrated in both theory and
practice that Tempo is self-tuning and robust for achieving guaranteed
SLOs in production database systems.

The current implementation of Tempo can simulate RMs like Mesos and
YARN efficiently, using the time warp mechanism. Tempo can leverage RM
simulators that have already been developed for several popular
systems such as Borg, Apollo, and Omega. However, most of these
simulators are designed to reproduce the real-time behavior of the RM,
which may not deliver comparable efficiency as time warp
simulation. Thus, one interesting direction is to add effective
support for other RMs in Tempo.

The SLO abstraction in Tempo, i.e., QS metrics, allows each tenant to
specify one or more SLOs, which apply to all workloads of the
tenant. The effective spectrum of the SLOs corresponds to the
structure of resource configurations in popular RMs like Mesos and
YARN, where parameters are grouped by tenants (also known as pools or
queues). To support more fine-grained SLOs among workloads of the same
tenant, one workaround is to create hierarchical tenants as used in
the Hadoop Capacity Scheduler. One future direction in Tempo is to
provide native fined-grained SLO support for workloads from the same
tenant.

A third direction is to explore scenarios where each tenant executes
workloads exhibiting a mix of statistical characteristics.  The
current optimization in Tempo exploits the observation that workloads
from the same tenant follow relatively fixed statistical
characteristics. This assumption alleviates us from the restriction of
having to observe historical executions of a newly-submitted job (both
recurring and ad-hoc). To support workloads with a mix of statistical
characteristics, one potential approach is to decompose the workloads
and then distribute the workloads to separate tenants.

{\footnotesize \bibliographystyle{acm}
\bibliography{ref}}

\end{document}